\documentclass[a4paper]{article}
\usepackage{latexsym,amsthm,amsmath,amssymb}

\theoremstyle{plain}
\newtheorem{theorem}{Theorem}
\newtheorem{claim}[theorem]{Claim}
\newtheorem*{theorem*}{Theorem}
\newtheorem{lemma}[theorem]{Lemma}
\newtheorem*{lemma*}{Lemma}
\newtheorem{corollary}[theorem]{Corollary}
\newtheorem{observation}[theorem]{Observation}
\newtheorem*{observation*}{Observation}

\newtheorem*{conjecture*}{Conjecture}

\theoremstyle{definition}

\usepackage[utf8]{inputenc}

\usepackage[english]{babel}
\usepackage{graphicx}
\usepackage{type1ec}
\usepackage[T1]{fontenc}
\usepackage{bbm}
\usepackage{mathtools}
\usepackage{amssymb}
\usepackage{amsfonts}
\usepackage{dsfont}
\usepackage{pifont}
\usepackage{url}
\usepackage{multirow}
\usepackage{multicol}
\usepackage{pbox}
\usepackage{fancybox}
\usepackage{colortbl}
\usepackage{soul}
\usepackage{color}
\usepackage{tikz}
\usepackage{tkz-berge}
\usepackage{rotating}
\usepackage{pdflscape}
\usepackage{afterpage}
\usepackage{comment}
\usepackage{etoolbox}
\usepackage{environ}
\usepackage{dirtytalk}
\usepackage{authblk}

\usepackage{algorithm}
\usepackage{algorithmic}
\usepackage{hyperref}

\usetikzlibrary{decorations,arrows,shapes,positioning}

\newenvironment{cproof}{\begin{proof}}{\end{proof}}

\newcommand*\conj[1]{{\bar{#1}}}

\usepackage{complexity}

\newcommand{\up}{\uparrow}
\newcommand{\down}{\downarrow}

\newcommand{\bigO}[1]{{\mathcal{O}\!\left(#1\right)}}
\newcommand{\bigOs}[1]{{\mathcal{O}^*\!\left(#1\right)}}

\newcommand{\pname}[1]{\textsc{#1}}

\newcommand{\inners}{1.2pt}
\newcommand{\outers}{1pt}

\newclass{\Hard}{hard}
\newclass{\Hness}{hardness}
\newcommand{\NPH}{\NP\text{-}\Hard}
\newcommand{\WH}{\W\textsc{[1]-}\Hard}
\newcommand{\WHness}{\W\textsc{[1]-}\Hard}
\newclass{\para}{para}
\newclass{\Complete}{complete}
\newclass{\Cness}{completeness}
\newcommand{\NPc}{\NP\textsc{-}\Complete}

\newfunc{\dist}{dist}
\newfunc{\girth}{girth}
\newfunc{\nd}{nd}
\newfunc{\YES}{YES}
\newfunc{\NOi}{NO}
\newfunc{\ff}{ff}
\newfunc{\dc}{dc}
\newfunc{\dcc}{d\overline{c}}
\newfunc{\ml}{ml}
\newfunc{\cf}{cf}
\newfunc{\tw}{tw}
\newfunc{\ext}{ext}
\newfunc{\pat}{pat}
\newcommand{\angled}[1]{\left\langle{#1}\right\rangle}

\newtoggle{proof_toggle}
\toggletrue{proof_toggle}

\NewEnviron{tproof} {
    \iftoggle{proof_toggle} {%
        \begin{proof}\BODY\end{proof}}
    {%
    }
}
\newcommand{\ceil}[1]{\left\lceil#1\right\rceil}
\newcommand{\floor}[1]{\left\lfloor#1\right\rfloor}

\title{Structural Parameterizations for Equitable Coloring\footnote{A previous version of this work was was published in the proceedings of LATIN2020~\cite{equitable_latin}.}}
\author[1]{Guilherme C. M. Gomes}
\author[1]{Matheus R. Guedes}
\author[1]{Vinicius F. dos Santos}
\date{}

\affil[1]{Departamento de Ci\^encia da Computa\c{c}\~{a}o, Universidade Federal de Minas Gerais -- Belo Horizonte, Brazil}

\begin{document} 
    \maketitle              
    \begin{abstract}
    An $n$-vertex graph is equitably $k$-colorable if there is a proper coloring of its vertices such that each color is used either $\floor{n/k}$ or $\ceil{n/k}$ times.
    While classic \textsc{Vertex Coloring} is fixed parameter tractable under well established parameters such as pathwidth and feedback vertex set, equitable coloring is \W[1]-\Hard.
    We present an extensive study of structural parameterizations of \textsc{Equitable Coloring}, tackling both tractability and kernelization questions.
    We begin by showing that the problem is fixed parameter tractable when parameterized by distance to cluster or by distance to co-cluster --- improving on the \FPT\ algorithm of Fiala et al. [Theoretical Computer Science, 2011] parameterized by vertex cover --- and also when parameterized by distance to disjoint paths of bounded length.
    To justify the latter result, we adapt a proof of Fellows et al. [Information and Computation, 2011] to show that \textsc{Equitable Coloring} is \W[1]-\Hard\ when simultaneously parameterized by distance to disjoint paths and number of colors.
    In terms of kernelization, on the positive side we present a linear kernel for the distance to clique parameter and a cubic kernel when parameterized by the maximum leaf number; on the other hand, we show that, unlike \pname{Vertex Coloring}, \pname{Equitable Coloring} does not admit a polynomial kernel when jointly parameterized by vertex cover and number of colors, unless $\NP \subseteq \coNP/\poly$.
    We also revisit the literature and derive other results on the parameterized complexity of the problem through minor reductions or other observations.
    \end{abstract}
    
    \section{Introduction}

\pname{Equitable Coloring} is a variant of the classical \pname{Vertex Coloring} problem: we want to partition an $n$ vertex graph into $k$ independent sets such that each of these sets has either $\floor{n/k}$ or $\ceil{n/k}$ vertices.
The smallest integer $k$ for which $G$ admits an equitable $k$-coloring is called the \emph{equitable chromatic number} of $G$. %
\pname{Equitable Coloring} was first discussed in~\cite{first_equitable}, with an intended application for municipal garbage collection, and later in processor task scheduling~\cite{mutual_exclusion_scheduling}, communication control~\cite{scheduling_traffic}, and server load balancing~\cite{domain_decomposition}.
Lih~\cite{equitable_survey} presented an extensive survey covering many of the results developed in the last 50 years.
Its focus, however, is not algorithmic, and most of the presented results are bounds on the equitable chromatic number for various graph classes.

Many complexity results for \pname{Equitable Coloring} arise from a related problem, known as \pname{Bounded Coloring}, as observed by Bodlaender and Fomin~\cite{equitable_treewidth}.
On \pname{Bounded Coloring}, we ask that the size of the independent sets be bounded by an integer $\ell$.
Among the positive results for \pname{Bounded Coloring}, the problem is known to be solvable in polynomial time for:
split graphs~\cite{equitable_split},
complements of interval graphs~\cite{graph_partitioning1}, complements of bipartite graphs~\cite{graph_partitioning1}, and
forests~\cite{mutual_exclusion_scheduling}.
Baker and Coffman~\cite{mutual_exclusion_scheduling} present the first algorithm for \pname{Bounded Coloring} on trees, while Jarvis and Zhou~\cite{equitable_trees} show how to compute an optimal $\ell$-bounded coloring of a tree through a novel characterization.
For cographs, bipartite and interval graphs, there are polynomial-time algorithms when the number of colors $k$ is fixed.
In terms of parameterized complexity, in \cite{equitable_treewidth} an \XP\ algorithm is given for \pname{Equitable Coloring} parameterized by treewidth, while Fiala et al.~\cite{equitable_vertex_cover} show that the problem is \FPT\ parameterized by vertex cover and Enciso et al.~\cite{equitable_max_leaf} show that it is \FPT\ parameterized by the maximum leaf number.
Recently, Gomes et al.~\cite{equitable_dmtcs} proved that, when parameterized by the treewidth of the complement graph, \pname{Equitable Coloring} is \FPT.
Reddy~\cite{equitable_threshold} proved that the problem is tractable when the parameter is the distance to \textit{connected} threshold graphs, even though the connected qualifier is omitted in the claim;
Another claim of~\cite{equitable_threshold} is that \pname{Equitable Coloring} admits a polynomial kernel when parameterized by distance to threshold and number of colors; the proof, however, only works for distance to \textit{connected} threshold graphs.
Our result in Section~\ref{sec:vc} shows that the original claim is false, i.e. there can be no polynomial kernel when the parameter is the distance to threshold graphs unless $\NP \subseteq \coNP/\poly$.

The main contributions of this work are complexity results on \pname{Equitable Coloring} for parameterizations that are weaker than vertex cover, in the sense that the parameters are upper bounded by the vertex cover number, and a cubic kernel for the maximum leaf number parameterization.
In particular, we show that \pname{Equitable Coloring} is fixed parameter tractable when parameterized by distance to cluster, distance to co-cluster, or by distance to disjoint paths of bounded length.
Not only are the parameters weaker, but also in the first case, the algorithm is slightly faster than the one previously known for vertex cover, as it does not rely on \pname{Integer Linear Programming}; the running time, however, is still of the order of $2^\bigO{{k\log k}}$.
On the negative side, we show that the combined parameterization distance to disjoint paths and number of colors is insufficient to guarantee tractability.
Along with some of the works discussed here and in Section~\ref{sec:negatives}, our results cover many branches of the known graph parameter hierarchy~\cite{gpp}.

\smallskip
\noindent \textbf{Our results.}
In this work, we conduct a systematic study on the complexity of structural parameterizations for \pname{Equitable Coloring}.
We begin by revisiting the literature in Section~\ref{sec:negatives} and showing how previous work can be adapted in order to results on the parameterized complexity of \pname{Equitable Coloring}.
Afterwards, our first technical contributions prove that the problem is fixed-parameter tractable when parameterized by distance to cluster and co-cluster; these algorithms are faster than the previously known \FPT\ algorithm of Fiala et al.~\cite{equitable_vertex_cover}.
Still in terms of tractability, we improve on the reduction of Fellows et al.~\cite{colorful_treewidth} and show that, when parameterized by the distance to disjoint paths and number of colors, \pname{Equitable Coloring} remains \WH.
We then turn our attention to kernelization, by first showing
that the problem admits a linear kernel when parameterized by distance to clique; we point out that the best known kernel for \pname{Vertex Coloring} under this parameterization is quadratic.
Afterwards, we present a cubic kernel when the max leaf number is the parameter.
Finally, we adapt the reduction of Section~\ref{sec:dpath} to show that, unless $\NP \subseteq \coNP/\poly$, \pname{Equitable Coloring} does not admit a polynomial kernel when jointly parameterized by vertex cover and number of colors; it is also worthy to note that \pname{Vertex Coloring} admits a polynomial kernel under vertex cover~\cite{data_reduction}.
Our contributions cover a large part of the graph parameter hierarchy~\cite{gpp}.
We list the question of whether or not the problem is tractable under feedback edge set as our main open problem.
    
\noindent \textbf{Notation and Terminology.}
    We use standard graph theory notation and nomenclature for our parameters, following classical textbooks in the areas~\cite{murty,cygan_parameterized}.
    Define $[k] = \{1,\dots, k\}$ and $2^S$ the \emph{powerset} of $S$.
    A \emph{$k$-coloring} $\varphi$ of a graph $G$ is a function $\varphi: V(G) \mapsto~[k]$.
    Alternatively, a $k$-coloring is a $k$-partition $V(G) \sim \{\varphi_1, \dots, \varphi_k\}$ such that $\varphi_i = \{u \in V(G) \mid \varphi(u) = i\}$.
    A $k$-coloring is said to be \emph{equitable} if, for every $i \in [k]$, $\floor{n/k} \leq |\varphi_i| \leq \ceil{n/k}$; it is \emph{proper} if every $\varphi_i$ is an independent set.
    Unless stated, all colorings are proper.
    The \pname{Equitable Coloring} problem asks whether or not $G$ can be equitably $k$-colored.
    A graph $G$ is a \textit{subdivision} of a graph $H$ if a graph isomorphic to $G$ can be obtained by replacing some edges of $H$ with paths of arbitrary length.
    A graph is a \textit{cluster graph} if each of its connected components is a clique; a \textit{co-cluster graph} is the complement of a cluster graph.
    The \textit{distance to cluster} (\textit{co-cluster}) of a graph $G$, denoted by $\dc(G)$ ($\dcc(G)$), is the size of the smallest set $U \subseteq V(G)$ such that $G - U$ is a cluster (co-cluster) graph.
    Using the terminology of~\cite{cai_split}, a set $U \subseteq V(G)$ is an $\mathcal{F}$-\textit{modulator} of $G$ if the graph $G - U$ belongs to the graph class $\mathcal{F}$.
    When the context is clear, we omit the qualifier $\mathcal{F}$.
    In particular, we say that $U$ is a $P_\ell$-modulator if each connected component of $G \setminus U$ has at most $\ell$ vertices and is a path.
    For cluster and co-cluster graphs, one can decide if $G$ admits a modulator of size $k$ in time \FPT\ on $k$~\cite{clusterFPT}.
    The \textit{maximum leaf number}, often called the max leaf number and denoted by $\ml(G)$, is the maximum number of leaves on some spanning forest of $G$ with as many connected components as $G$~\cite{max_subdivision}.
    \section{Literature corollaries and minor observations}\label{sec:negatives}

The original $\NPc$ results of Bodlaender and Jansen~\cite{graph_partitioning1}, despite being initially regarded as polynomial reductions for \pname{Bounded Coloring}, are a nice source of parameterized hardness.
To adapt their proofs to show that \pname{Equitable Coloring} parameterized by the number of colors is \WH\ on cographs and \para\NPH\ on bipartite graphs, it suffices to consider the version of \pname{Bin-Packing} where each bin must be completely filled for the first case, while the latter follows immediately since they prove that \pname{Bounded 3-Coloring} is \NPH\ on bipartite graphs; these imply that adding the distance to theses classes in the parameterization yields no additional power whatsoever.
Fellows et al.~\cite{colorful_treewidth} show that \pname{Equitable Coloring} parameterized by treewidth and number of colors is \WH, while an \XP\ algorithm parameterized by treewidth is given for both \pname{Equitable Coloring} and \pname{Bounded Coloring} by Bodlaender and Fomin~\cite{equitable_treewidth}.
In fact, the reduction shown in~\cite{colorful_treewidth} prove that, even when simultaneously parameterized by feedback vertex set, treedepth, and number of colors, \pname{Equitable Coloring} remains \WH.
Gomes et al.~\cite{equitable_dmtcs} show that the problem parameterized by number of colors, maximum degree and treewidth is \WH\ on interval graphs.
However, their intractability statement can be strengthened to number of colors and \textit{bandwidth}, with no changes to the reduction.
In~\cite{cai_split}, Cai proves that \pname{Vertex Coloring} is \WH\ parameterized by distance to split.

We also reviewed results on parameters weaker than distance to clique, since both distance to cluster and co-cluster fall under this category.
For minimum clique cover, we resort to the classic result of Garey and Johnson~\cite{garey_johnson} that \pname{Partition into Triangles} is \NPH.
By definition, a graph $G$ can be partitioned into vertex-disjoint triangles if and only if its complement graph can be equitably $(n/3)$-colored.
The reduction given in~\cite{garey_johnson} is from \pname{Exact Cover by 3-Sets}, and their gadget (which we reproduce in Figure~\ref{fig:triangles}) has the nice property that the complement graph $\overline{G}$ has a trivial clique cover of size nine: it suffices to pick one gadget $i$ and one clique for each $a_i^j$.
Thus, we have that \pname{Equitable Coloring} is \para\NPH\ parameterized by minimum clique cover.
To see that when also parameterizing by the number of colors there is an \FPT\ algorithm, we first look at the parameterization maximum independent set $\alpha$ and number of colors $k$, both of which we assume to be given on the input.
First, if $k\alpha < n$, the instance is trivially negative, so we may assume $k\alpha \geq n$; but, in this case, we can spend exponential time on the number of vertices and still run in \FPT\ time.
Finally, we reduce from \pname{Equitable Coloring} parameterized by the number of colors $k$ to \pname{Equitable Coloring} parameterized by $k$ and minimum dominating set.
If we take the source graph $G$ and add $\frac{n}{k}$ vertices $D = v_i, \dots, v_{\frac{n}{k}}$ with $N(v_i) = V(G)$ for all $v_i \in D$, the set $\{v_1, u\}$, with $u \in V(G)$, is a dominating set of the resulting graph $G'$; moreover, $G$ has an equitable $k$-coloring if and only if $G'$ is equitably $(k+1)$-colorable, thus proving that \pname{Equitable Coloring} parameterized by $k$ and minimum dominating set is \para\NPH.
A summary of the results discussed in this work is displayed in Figure~\ref{fig:diagram}.

\begin{figure}[!htb]
    \centering
    \begin{tikzpicture}[yscale=0.8, xscale=0.7]
         \GraphInit[unit=3,vstyle=Normal]
         \SetVertexNormal[Shape=circle, FillColor=black, MinSize=3pt]
         \tikzset{VertexStyle/.append style = {inner sep = \inners, outer sep = \outers}}
         \SetVertexLabelOut
         \Vertex[x=-4, y =-3, Lpos=270, Math, L={x_i}]{x}
         \Vertex[x=0, y =-3, Lpos=270, Math, L={y_i}]{y}
         \Vertex[x=4, y =-3, Lpos=270, Math, L={z_i}]{z}
         
         \Vertex[x=-5, y=-2, Lpos=225, Math, Ldist=-1pt, L={a_i^1}]{ax}
         \Vertex[x=-3, y=-2, Lpos=315, Math, Ldist=-1pt, L={a_i^2}]{bx}
         \Edges(ax,x,bx,ax)
         
         \Vertex[x=-1, y=-2, Lpos=225, Math, Ldist=-1pt, L={a_i^4}]{ay}
         \Vertex[x=1, y=-2, Lpos=315, Math, Ldist=-1pt, L={a_i^5}]{by}
         \Edges(ay,y,by,ay)
         
         \Vertex[x=3, y=-2, Lpos=225, Math, Ldist=-1pt, L={a_i^7}]{az}
         \Vertex[x=5, y=-2, Lpos=315, Math, Ldist=-1pt, L={a_i^8}]{bz}
         \Edges(az,z,bz,az)
         
         \Vertex[x=-2.5, y=0, Lpos=135, Math, Ldist=-1pt, L={a_i^3}]{a3}
         \Vertex[x=0, y=-1, Lpos=90, Math, Ldist=-1pt, L={a_i^6}]{a6}
         \Vertex[x=2.5, y=0, Lpos=45, Math, Ldist=-1pt, L={a_i^9}]{a9}
         
         \Edges(ax,a3,bx)
         \Edges(ay,a6,by)
         \Edges(az,a9,bz)
         \Edges(a3,a6,a9,a3)
         
    \end{tikzpicture}
        
    \caption{\pname{Exact Cover by 3-Sets} to \pname{Partition into Triangles} gadget of \cite{garey_johnson} representing the set $C_i = \{x_i,y_i,z_i\}$.}
    \label{fig:triangles}
\end{figure}
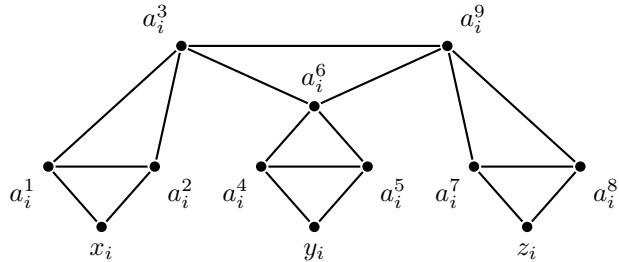

\begin{figure}[!htb]
    \hspace{-1.2cm}
    \includegraphics[scale=0.5]{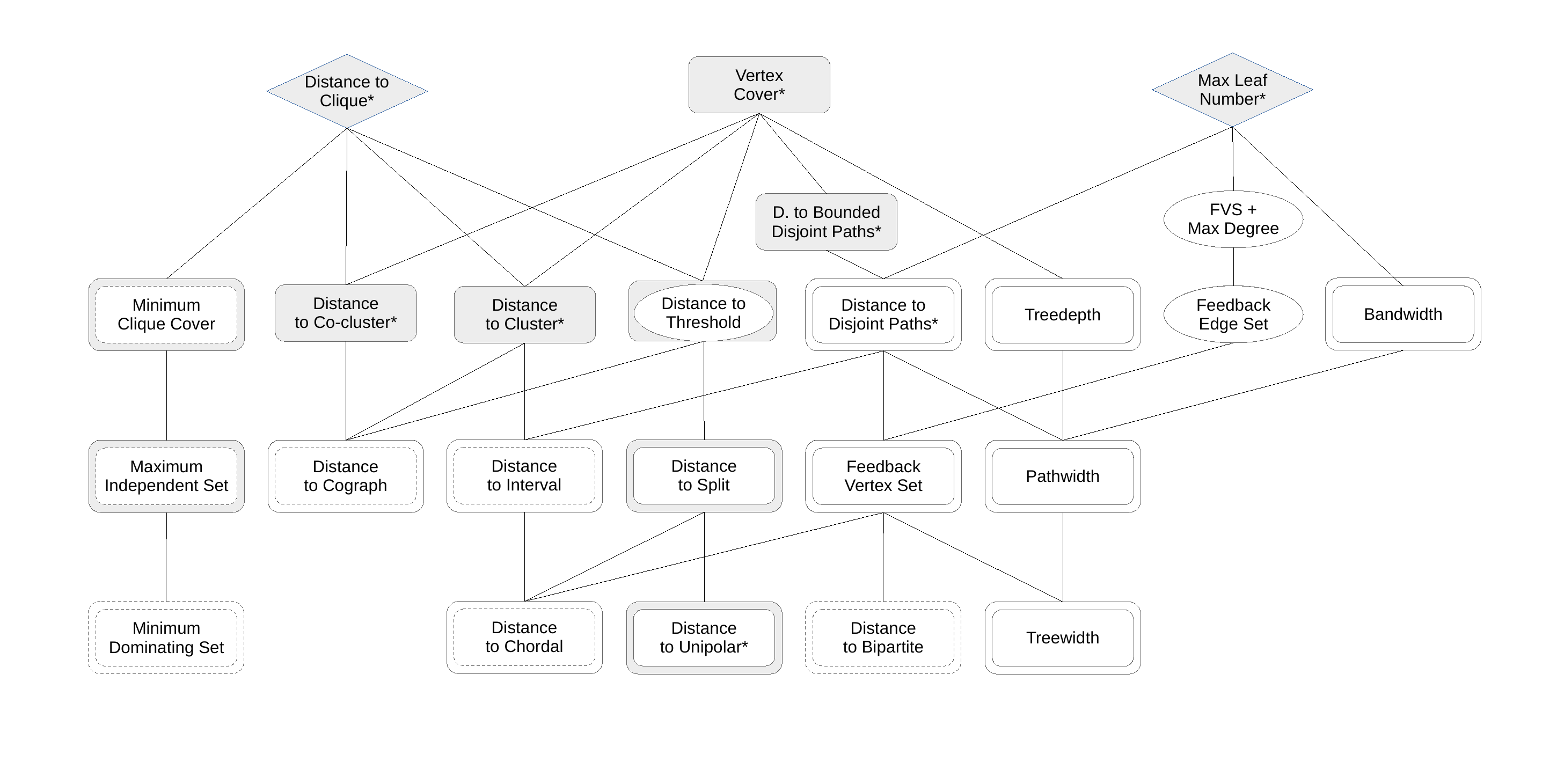}
    \caption{Hasse diagram of the parameterizations of \pname{Equitable Coloring} and their complexities.
    A single shaded box indicates that the problem is \FPT, while a shaded rhombus highlights a case that admits a polynomial kernel; two solid boxes represent \WHness\ even if also parameterized by the number of colors; if the inner box is dashed, the problem is \para\NPH; if the outer box is solid and shaded, additionally using the number of colors results in an \FPT\ algorithm; if it is not shaded, it remains \WH. Entries with a * are our main contributions. Ellipses mark open cases.}
    \label{fig:diagram}
\end{figure}
    \section{Equitable coloring parameterized by distance to cluster}
    \label{sec:dc}
    The goal of this section is to prove that \pname{Equitable Coloring} can be solved in \FPT\ time when parameterized by the distance to cluster of the input graph.
    As a corollary of this result, we show that unipolar graphs -- the class of graphs that have a clique as a modulator -- can be equitably $k$-colored in polynomial time.
    Throughout this section, we denote the modulator by $U$, the connected components of $G - U$ by $\mathcal{C} = \{C_1, \dots, C_r\}$, and define $\ell = \floor{\frac{n}{k}}$.
    
    The central idea of our algorithm is to guess one of the possible $|U|^{|U|}$ colorings of the modulator and extend this guess to the clique vertices using max-flow.
    First, given $U$, $\mathcal{C}$, and a coloring $\varphi'$ of $U$, we build an auxiliary graph $H$ as follows:
    $V(H) = \{s, t\} \cup A \cup W \cup V(G) \setminus U$, where $A = \{a_1, \dots, a_k\}$ represents the colors we may assign to vertices, $W = \{w_{ij} \mid i \in [k], j \in [r]\}$ whose role is to maintain the property of the coloring, $s$ is the source of the flow, $t$ is the sink of the flow, and $V(G) = \{v_1, \dots, v_n\}$ are the vertices of $G$.
    For the arcs, we have $E(H) = S \cup F_0 \cup F_1 \cup R \cup T$, where $S = \{(s, a_i) \mid i \in [k]\}$, $F_0 = \{(a_i, t) \mid i \in [k]\}$, $F_1 = \{(a_i, w_{ij}) \mid i \in [k], j \in [r]\}$, $R = \{(w_{ij}, v_p) \mid v_p \in C_j, N(v_p) \cap \varphi'_i = \emptyset\}$, and $T = \{(v_i, t) \mid v_i \in V(G) \setminus U\}$.
    As to the capacity of the arcs, we define $c : E(H) \rightarrow \mathbb{N}$, with $c(e \in S) = \ell$, $c((a_i, t)) = |\varphi'_i \cap U|$ and $c(e \in F_1 \cup R \cup T) = 1$.
    Semantically, the vertices of $A$ correspond to the $k$ colors, while each $w_{ij}$ ensures that cluster $C_j$ has at most one vertex of color $i$.
    Regarding the arcs, $F_0$ corresponds to the initial assignment of colors to the vertices of the modulator, and $R$ encodes the adjacency between vertices of the clusters and colored vertices of the modulator.
    Note that the arcs in $F_0$ and $R$ are the only ones affected by the pre-coloring $\varphi'$.
    An example of the constructed graph can be found in Figure~\ref{fig:cluster_ex}.
    
    
    \begin{figure}[!htb]
        \centering
            \begin{tikzpicture}[scale=0.7]
                \GraphInit[unit=3,vstyle=Normal]
                \SetVertexNormal[Shape=circle, FillColor=white, MinSize=3pt]
                \tikzset{VertexStyle/.append style = {inner sep = \inners, outer sep = \outers}}
                \SetVertexLabelOut
                \begin{scope}
                    \begin{scope}
                        \Vertex[x=0, y=0, Lpos=180,Math]{v_1}
                        \Vertex[x=2, y=0, Lpos=0, Math]{v_2}
                        \Vertex[x=0, y=2, Lpos=90,Math]{v_3}
                        \Vertex[x=2, y=2, Lpos=90,Math]{v_4}
                        \Vertex[x=0, y=-2, Lpos=-90,Math]{v_5}
                        \Vertex[x=2, y=-2, Lpos=-90,Math]{v_6}
                        \Edges(v_1,v_2,v_4,v_3,v_1,v_5,v_6,v_2)
                    \end{scope}
                    \draw (-0.2, -0.2) rectangle (2.2, 0.2);
                    \begin{scope}
                        \tikzset{VertexStyle/.append style = {shape = circle, inner sep = 2pt}}
                        \AddVertexColor{black}{v_1}
                    \end{scope}
                    \begin{scope}
                        \tikzset{VertexStyle/.append style = {shape = rectangle, inner sep = 3pt}}
                        \AddVertexColor{black}{v_2}
                    \end{scope}
                \end{scope}
            \end{tikzpicture}
            \hfill
            \begin{tikzpicture}[scale=0.7]
                \GraphInit[unit=3,vstyle=Normal]
                \SetVertexNormal[Shape=circle, FillColor=white, MinSize=3pt]
                \tikzset{VertexStyle/.append style = {inner sep = \inners, outer sep = \outers}}
                \SetVertexLabelOut
                \begin{scope}
                    \Vertex[x=0,y=0,Lpos=180, Math]{s}
                    \begin{scope}
                        \tikzset{VertexStyle/.append style = {shape = circle, inner sep = 2pt}}
                        \Vertex[x=1, y=1, Lpos=90, Math]{a_1}
                        \AddVertexColor{black}{a_1}
                    \end{scope}
                    \begin{scope}
                        \tikzset{VertexStyle/.append style = {shape = rectangle, inner sep = 3pt}}
                        \Vertex[x=1, y=-1, Lpos=-90, Math]{a_2}
                        \AddVertexColor{black}{a_2}
                    \end{scope}
                    \Vertex[x=7, y=0, Lpos=0, Ldist=1pt, Math]{t}
                    \foreach \i in {1,2} {
                        \Edge[style={->, dashed}](s)(a_\i)
                        \foreach \j in {1,2} {
                            \pgfmathsetmacro{\y}{-3*(\i-1) + 1.5 - (1*(\j-1) - 0.5)}
                            \pgfmathtruncatemacro{\id}{2*(\i-1) + \j-1 + 3}
                            \pgfmathtruncatemacro{\lpos}{180*(\i-1) + 90}
                            \Vertex[x=3, y=\y, Lpos=\lpos, Math, L={w_{\i\j}}]{w_\i\j}
                            \Vertex[x=5, y=\y, Lpos=\lpos, Math, L ={v_\id}]{cv_\id}
                            \Edge[style={->}](a_\i)(w_\i\j)
                            \Edge[style={->}](cv_\id)(t)
                        }
                    }
                    \Edge[style={->, dashed, bend right=10}](a_1)(t)
                    \Edge[style={->, dashed, bend left=10}](a_2)(t)
                    \Edge[style={->}](w_21)(cv_3)
                    \Edge[style={->}](w_11)(cv_4)
                    \Edge[style={->}](w_22)(cv_5)
                    \Edge[style={->}](w_12)(cv_6)
                    
                \end{scope}
        \end{tikzpicture}
    \caption{{(left)} The input graph with $U = \{v_1, v_2\}$; {(right)} Auxiliary graph constructed from the precoloring of $U$. Solid arcs have unit capacity.\label{fig:cluster_ex}}
    \end{figure}
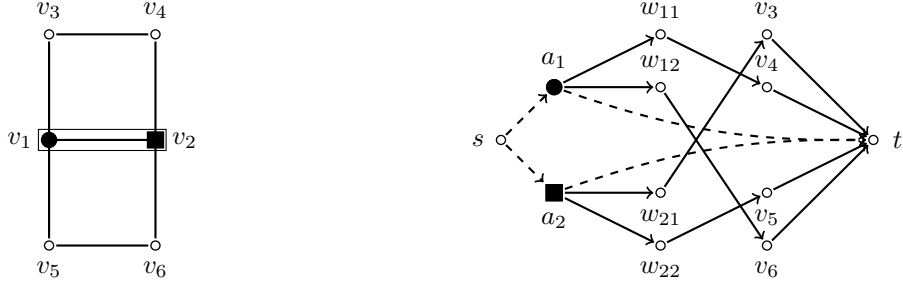

    Now, let $f : E(H) \rightarrow \mathbb{N}$ be the function corresponding to the max-flow from $s$ to $t$ obtained using any of the algorithms available in the literature~\cite{maxflow}.
    Our first observation, as given by the following lemma, is that, if no $(s,t)$-flow saturates the outbound arcs of $s$, then we cannot extend $\varphi'$ to equitably $k$-color $G$.
    
    \begin{lemma}\label{lem:partial_flow}
        If there is some $e \in S$ with $f(e) < \ell$, then $G$ does not admit an equitable $k$-coloring that extends $\varphi'$.
    \end{lemma}

    \begin{proof}
        By contraposition, suppose that $G$ admits an equitable $k$-coloring $\varphi$ extending $\varphi'$.
        As such, there exists some coloring $\psi$ satisfying $\varphi' \subset \psi \subset \varphi$ such that for all $i \in [k]$ we have $|\psi_i| = \ell$.
        To construct function $f$, begin by setting $f(e) = \ell$, for each $e \in S$.
        If $e = (a_i, t)$, i.e., $e \in F_0$, $f(e) = |\psi_i \cap U|$.
        For each $e \in T$ such that $e = (v, t)$ and $v$ is colored by $\psi$, set $f(e) = 1$.
        For each $i \in [k], j \in [r]$, if $\psi_i \cap C_j \neq \emptyset$ then we set $f((a_i, w_{ij})) = 1$ and $f((w_{ij}, v)) = 1$, where $v$ is the only element of $\psi_i \cap C_j$.
        Note that the arc $(w_{ij}, v)$ exists since $v \in V(G) \setminus U$ belongs to clique $C_j$, has no neighbor with color $i$ in $U$, and $\psi$ is a coloring of $G$.
        All other arcs have $f(e) = 0$.
        To see that $f$ corresponds to a feasible flow of $H$ under $c$, first note that $f(e) \leq c(e)$ for every $e \in E(H)$ follows from the construction of $f$.
        To conclude the proof, it suffices to check that $\sum_{j \in [r]} f((a_i, w_{ij})) = \ell - |\psi \cap U|$, which must be the case, since $f((a_i,w_{ij})) = 1$ if and only if there is some $v \in C_j \cap \psi_i$.
    \end{proof}
    
    We may now assume that $f(e) = c(e)$ for every $e \in S$.
    Let $c'(e \notin S) = c(e)$, $c'(e \in S) = c(e) + 1$.
    We resume the search for augmenting paths on the network, replacing $c(\cdot)$ with $c'(\cdot)$, until it stops and returns the maximum $(s,t)$-flow $g$.
    
    \begin{lemma}\label{lem:flow_lower_bound}
        For every $e \in S$, $g(e) \geq f(e)$.
    \end{lemma}

    \begin{proof}
        Suppose that $g(e) < f(e)$ for some $e \in S$, than at some step the max-flow algorithm must have picked an augmenting path with an arc $e = (a_i,s)$, but this leads to a contradiction as the augmented path would not be simple. 
    \end{proof}
    
    \begin{lemma}
        \label{lem:col_construction}
        The maximum ($s,t$)-flow $F$ given by $g$ is equal to the number of vertices of $G$ if and only if there is an equitable $k$-coloring of $G$ that extends $\varphi'$.
    \end{lemma}
    
     \begin{proof}
         Suppose that the maximum ($s,t$)-flow is $|V(G)|$.
         Reading the flow function to retrieve the coloring $\varphi$ is straightforward: for each $(w_{ij}, v_p) \in R$ with $f((w_{ij}, v_p)) = 1$, color $v_p$ with $i$.
         Since there is only one arc $(v, t)$ with unit capacity, there is an unique $i$ such that $v \in \varphi_i$.
         By Lemma~\ref{lem:flow_lower_bound}, $\ell \leq |\varphi_i|$ and, since the capacity of each outbound arc of $s$ is $\ell + 1$, $|\varphi_i| \leq \ell + 1$.
         That $\varphi_i$ is a proper coloring follows from the hypothesis that $\varphi'$ was a proper coloring of $G[U]$,  the arcs in $R$ encode the constraints that a vertex $v \in C_j$ cannot be colored with any of the colors $\{i \mid N(v) \cap \varphi'_i \neq \emptyset\}$, and that only one vertex $v \in C_j$ satisfies $f((w_{ij}, v)) = 1$ since $c'((a_i, w_{ij})) = 1$.
         Finally, $\varphi$ extends $\varphi'$ since no vertex of $U$ was recolored.
         For the converse, take an equitable $k$-coloring $\varphi$ that extends $\varphi'$ and define the flow function as in the proof of Lemma~\ref{lem:partial_flow}, but with $g((s, a_i)) = |\varphi_i|$.
         The same arguments hold, concluding the proof.
     \end{proof}
    
    At this point we are essentially done.
    Lemmas~\ref{lem:partial_flow} and~\ref{lem:col_construction} guarantee that, if the max-flow algorithm fails to yield a large enough flow, a fixed pre-coloring of $U$ cannot be extended; moreover, the latter also implies that, if an extension is possible, max-flow correctly finds it.
    
    \begin{theorem}
        \label{thm:dc_fpt}
        \pname{Equitable Coloring} parameterized by distance to cluster can be solved in \FPT\ time.
    \end{theorem}
    
    \begin{proof}
        Suppose we are given $U$.
        Now, for each of the $\bigO{B_{|U|}}$ possible colorings of $U$ --- where $B_n$ is the $n$-th Bell number ---  construct $H$ and execute the algorithm described in this section.
        Since max-flow can be solved in polynomial time~\cite{maxflow} and we have an \FPT\ number of colorings of $U$, we have an \FPT\ algorithm for \pname{Equitable Coloring}.
    \end{proof}
    
    It is worthy to note here that there is nothing special about the capacities of the arcs in $S$; they act only as upper bounds to the number of vertices a color may be assigned to.
    Thus, not surprisingly, the same algorithm applies to problems where the size of each color class is only upper bounded.
    This will be particularly useful in the next session.
    Looking at the proof of Theorem~\ref{thm:dc_fpt}, the only non-polynomial step is guessing the coloring of the modulator.
    A straightforward corollary is that if there is a polynomial number of distinct colorings of $U$ and this family can be computed in polynomial time, we can apply the same ideas and check if an equitable $k$-coloring of the input graph exists in polynomial time.
    In particular, unipolar graphs satisfy the above condition.
    If we parameterize by distance to unipolar the problem remains \WH\ due to the hardness for split graphs.
    On the other hand, if we parameterize by distance to unipolar $d$ and the number of colors $k$ we have an \FPT\ algorithm: the central clique of $G - U$ has at most $k$ vertices, so we can treat $G$ as a graph with distance to cluster at most $k + d$ and apply Theorem~\ref{thm:dc_fpt}.
    
    \begin{corollary}
        \pname{Equitable Coloring} on unipolar graphs is in \P.
        When parameterized by distance to unipolar, the problem remains \WH; if also parameterized by the number of colors, it is solvable in \FPT\ time.
    \end{corollary}
    
    We note that this corollary does not contradict the \WHness\ proof of Gomes et al.~\cite{equitable_dmtcs} for \pname{Equitable Coloring} on disjoint union of split graphs when parameterized by number of colors, since the class of split graphs is not closed under disjoint union.
    \section{Distance to co-cluster} 
    \label{sec:dcc}
    Before proceeding to our hardness results, we discuss an \FPT\ algorithm when parameterized by distance to co-cluster.
    Interestingly, the key ingredient to our approach is the algorithm presented in Section~\ref{sec:dc}, which we use to compute the transitions between states of our dynamic programming table.
    Much like in the previous section, we denote by $U$ the set of vertices such that $G - U$ is a co-cluster graph, and by $\mathcal{I} = \{I_1, \dots, I_r\}$ the independent sets of $G - U$.
    The following observation follows immediately from the fact that $G - U$ is a complete $r$-partite graph; it allows us to color the sets of $\mathcal{I}$ independently.
    
    \begin{observation}
        \label{obs:containment}
        In any $k$-coloring $\varphi$ of $G$, for every color $i$, there is at most one $j \in [r]$ such that $\varphi_i \cap I_j \neq \emptyset$.
    \end{observation}
    
    Suppose we are already given $U$, a coloring $\psi$ of $U$, and the additional restriction that colors $P \subseteq \psi(U)$ must be used on $\ell + 1$ vertices.
    We index our dynamic programming table by $(S, p, q, j)$, where $S \subseteq \psi(U)$ stores which colors of the modulator still need to be extended, $p$ is the number of colors not in $\psi(U)$ that must still be used $\ell+1$ times, $q$ the number of colors not in $\psi(U)$ that must still be used on $\ell$ vertices, and $j \in [r]$ indicates which of the independent sets we are trying to color.
    Our goal is to show that $f_{\psi, P}(S, p, q, j) = 1$ if and only if there is a coloring of $G_j = G[U \bigcup_{i = j}^r I_i]$ respecting the constraints given by $(S, p, q, j)$.
    Intuitively, guessing $P$ allows us to index the table by the number of colors not in $\psi(U)$ according only to the capacity of each color, otherwise it would be significantly harder to know, at any time in the algorithm, how many colors should be used on $\ell + 1$ vertices.
    To compute $f_{\psi, P}(S, p, q, j)$, we essentially test every possibility of extension of the colors in $S$ that respects the constraint imposed by $P$ and allows the completion of the coloring of the $j$-th independent set of $G - U$.
    Because of Observation~\ref{obs:containment}, the colors not in $\psi(U)$ are confined to a single independent set and, thus, it suffices to consider only how many colors of size $\ell + 1$ we are going to use in $I_j$.
    We implement this transitioning according to Equation~\ref{eq:transition}:
    
    \begin{equation}
        \label{eq:transition}
        f_{\psi, P}(S, p, q, j) = \max_{(R,x,y)\ \in\ \ext(S, p, q, j)} f_{\psi, P \setminus R}(S \setminus R, p - x, q - y, j+1)
    \end{equation}
    
    \noindent where $\ext(S, p, q, j)$ is the set of all triples $(R, x, y)$, with $R \subseteq S$, such that 
    each color $i \in R$ can be extended to $I_j$, while $x$ and $y$ satisfy the system:
    
    \[\begin{cases}
        x(\ell + 1) + y \ell = |I_j| - \alpha_j \\
        0 \leq x \leq p \\
        0 \leq y \leq q
    \end{cases}\]
    
    \noindent where $\alpha_j$ is the number of vertices of $I_j$ used to extend the colors of $R$ to $I_j$.
    Note that $|\ext(S, p, q, j)| \leq 2^{|S|}n$, so it holds that, for each fixed $\psi$ and $P$, our dynamic programming table can be computed in $\bigOs{3^{|U|}}$ time if and only if we can compute $\ext(S,p,q,j)$ in $\bigOs{|\ext(S,p,q,j)|}$ time.
    
    \begin{lemma}\label{lem:ext}
        $\ext(S,p,q,j)$ can be computed in $\bigOs{|\ext(S,p,q,j)|}$ time.
    \end{lemma}

    \begin{proof}
        Given a triple $(R, x, y)$, with $R \subseteq S$, satisfying the above conditions, it suffices to show that membership in $\ext(S,p,q,j)$ can be decided in polynomial time.
        Actually, the challenge is to determine whether or not the colors in $R$ can be extended to include $I_j$.
        In order to do so, let $U' = \{v \in U \mid \psi(v) \in R\}$, and $G'$ be the subgraph of $G$ induced by the vertices $I_j \cup U'$.
        Observe that $U'$ is actually a vertex cover for $G'$; in particular, it is a modulator to cluster graph of size bounded by $\dcc(G)$.
        As such, we can interpret the task as follows: is there an induced subgraph of $G'$ that be colored with $|R|$ colors, that extends $\psi$, the colors in $P \cap R$ are used exactly $\ell + 1$ times and the others $P \setminus R$ $\ell$ times?
        This can be answered with a slight modification of the algorithm given in Section~\ref{sec:dc}: after the initial max-flow is executed, i.e. if the converse of Lemma~\ref{lem:partial_flow} holds, instead of updating the capacity of every arc leaving the source to $\ell+1$ we only update those corresponding to the colors in $P$, the same argumentation holds.
    \end{proof}
    
    \begin{lemma}\label{lem:dp_dcc}
        $f_{\psi, P}(S,p,q,j) = 1$ if and only if $\psi$ can be extended to a coloring $\varphi$ of $G_j$ using the colors of $S$, with each color in $P$ used in $\ell + 1$ vertices, $p$ extra color classes of size $\ell+1$, and $q$ color classes of size $\ell$.
    \end{lemma}
    
    \begin{proof}
        For the forward direction, we proceed by induction.
        The base case, where $j = r + 1$ is equal to the modulator, we can define $f_{\psi, P}(S,p,q,r+1) = 1$ if and only if $P = S = \emptyset$, and $p = q = 0$. 
        Now, take an arbitrary entry of the table with $i < r+1$; by Equation~\ref{eq:transition}, $f_{\psi, P}(S, p, q, j) = 1$ if and only if there is some triple $(R, x ,y)$ of $\ext(S, p, q, j)$ with $f_{\psi, P \setminus R}(S \setminus R, p - x, q - y, j + 1) = 1$.
        That is, there is a coloring of $G_{j+1}$ under the constraints imposed by $P \setminus R$ and tuple $(S \setminus R, p - x, q - y, j + 1)$.
        Moreover, since $V(G_j) \setminus V(G_{j+1}) = I_j$, no color of $R$ is used in any vertex of $V(G_{j+1}) \setminus U$, and by the definition of $\ext(S, p, q, j)$, it is possible to extend the coloring of $G_{j+1}$ to include the vertices of $I_j$ while keeping it proper and respecting the constraints imposed by $P$.
        
        Conversely, take a coloring $\varphi$ of $G_j$ satisfying the hypothesis, define $R = \varphi(U) \cap \varphi(I_j)$, and let $\varphi'$ be the restriction of $\varphi$ to $G_{j+1}$.
        By Observation~\ref{obs:containment}, every color in $R \subseteq S$ is used only on vertices of $U \cup I_j$, say $\alpha_j$ vertices of $I_j$, and there must be integers $0 \leq x \leq p$, $0 \leq y \leq q$ satisfying $(\ell + 1)x + \ell y + \alpha_j = |I_j|$ - that is, there are $x$ colors of size $\ell+1$ and $y$ colors of size $\ell$ used exclusively on the remaining vertices of $I_j$.
        By definition, the triple $(R, x, y)$ belongs to $\ext(S, p, q, j)$ and if $f_{\psi, P \setminus R}(S \setminus R,p - x, q - y, j + 1) = 1$ we are done.
        By induction on the number of available colors, this assertion holds since $|R| + x + y \geq 1$, otherwise $I_j = \emptyset$.
    \end{proof}
    
    Finally, all that is left is to show that the number of colorings of $U$ and the constraint set $P$ can both be computed in \FPT\ time.
    
    \begin{theorem}\label{thm:dcc_fpt}
        \pname{Equitable Coloring} can be solved in \FPT\ time when parameterized by distance to co-cluster.
    \end{theorem}
    
    \begin{proof}
        We can guess all possible colorings $\psi$ of $|U|$ in time $\bigOs{B_{|U|}}$ and for each color of $\psi$ we need to choose between adding it to $P$ or not. So we have $\bigOs{B_{|U|}2^{|U|}}$ cases. By Lemmas~\ref{lem:ext} and~\ref{lem:dp_dcc} each case can be solved in \FPT\ time parameterized by \dcc, so our algorithm is \FPT\ parameterized by \dcc.
    \end{proof}
    
    It is important to note that the above algorithm does not contradict the $\NP$-$\Hness$ of \pname{Equitable Coloring} on bipartite graphs, since solving the problem on \textit{complete} bipartite graphs is in \P.
    Moreover, if $U = \emptyset$, all steps of the algorithm are performed in polynomial time, yielding the following corollary.
    
    \begin{corollary}
        \pname{Equitable Coloring} of complete multipartite graphs is in \P.
    \end{corollary}

    \section{Distance to Disjoint Paths}
\label{sec:dpath}

We now investigate \pname{Equitable Coloring} parameterized by the \textit{distance to disjoint paths}, which is upper bounded by vertex cover and lower bounded by feedback vertex set.
Contrary to our expectations, we show that the problem is $\W[1]$-$\Hard$ even if we also parameterize by the number of colors.
To accomplish this, we make use of two intermediate problems, namely \pname{Number List Coloring} and \pname{Equitable List Coloring} parameterized by the number of colors.
The latter is very similar to \pname{Equitable Coloring} but to each vertex $v$ is assigned a list $L(v) \subseteq [k]$ of admissible colors.
\pname{Number List Coloring} generalizes it in the sense that now we are given a function $h : [k] \mapsto \mathbb{N}$ and color $i$ must be used \emph{exactly} $h(i)$ times.
As a first step, we show that \pname{Number List Coloring} parameterized by the number of colors is $\W[1]$-$\Hard$ on paths.
By roughly doubling the number of colors and vertices used in the construction of~\cite{colorful_treewidth}, we are able to use, essentially, the same arguments.
The source problem is \pname{Multicolored Clique} parameterized by the solution size $k$: given a graph $H$ such that $V(H)$ is partitioned in $k$ color classes $\mathcal{V} = \{V_1, \dots, V_k\}$, we want to determine if there is a $k$-colored clique in $H$.
We denote the edges between $V_i$ and $V_j$ by $E(V_i, V_j)$, $|V(H)|$ by $n$, and $|E(H)|$ by $m$.
We may assume that $|V_i| = N$ and $|E(V_i, V_j)| = M$ for every $i,j$; to see why this is possible, we may take $k!$ disjoint copies of $H$, each corresponding to a permutation of the color classes and, for each edge $uv \in E(H)$, we connect each copy of $u$ to each copy of $v$.
In our reduction, we interpret a clique as a set of \textit{oriented} edges between color classes, i.e. an edge $e \in E(V_i, V_j)$ is selected \textit{twice}: once from $V_i$ to $V_j$, and once from $V_j$ to $V_i$.
As such, we have \textit{two} gadgets for each edge of $H$.

\medskip\noindent \textbf{Construction.}
Due to the list nature of the problem, we assign semantic values to each set of colors.
In our case, we separate them in four types:

\begin{itemize}
    \item[Selection:] The colors $\mathcal{S} = \{\sigma(i,j) \mid (i,j) \in [k]^2, i \neq j\}$ and $\mathcal{S'} = \{\sigma'(i,j) \mid (i,j) \in [k]^2, i \neq j\}$ are used to select which edges must belong to the clique.
    \item[Helper:] $\mathcal{Y}$ and $\mathcal{X}$ satisfy $|\mathcal{Y}| = |\mathcal{X}| = |\mathcal{S}|$. These two sets of colors force the choice made at the root of the edge gadgets to be consistent across the gadget.
    \item[Symmetry:] The colors $\mathcal{E} = \{\varepsilon(i,j) \mid (i,j) \in [k]^2, i < j\}$  and $\mathcal{E'} = \{\varepsilon'(i,j) \mid (i,j) \in [k]^2, i < j\}$ guarantee that, if edge $e \in E(V_i, V_j)$ is picked from  $V_i$ to $V_j$, it must also be picked from $V_j$ to $V_i$.
    \item[Consistency:] Colors $\mathcal{T} = \{\tau_i(r,s) \mid i \in [k],\ r,s \in [k] \setminus \{i\}, \ r < s\}$ and $\mathcal{T'} = \{\tau'_i(r,s) \mid i \in [k],\ r,s \in [k] \setminus \{i\}, \ r < s\}$ ensure that if the edge $uv$ is chosen between $V_i$ and $V_j$, the edge between $V_i$ and $V_r$ must also be incident to $u$.
\end{itemize}

Before detailing the gadgets themselves, we define what is, in our perception, one of the most important pieces of the proof.
For each vertex $v \in V(H)$, choose an arbitrary but unique integer in the range $[n^2 + 1, n^2 + n]$ and, for each edge $e$, a unique integer in the range $[2n^2 + 1, 2n^2 + m]$.
These are the \textit{up-identification} numbers of vertex $v$ and edge $e$, denoted by $v_\up$ and $e_\up$, respectively.
Now, choose a suitably huge integer $Z$, say $n^3$, and define the \textit{down-identification number} for $v$ as $v_\down = Z - v_\up$.
These quantities play a key role on the numerical targets for the symmetry and consistency colors; since they are unique, these identification numbers tie together the choices between edge gadgets.

For each pair $i,j \in [k]$, with $i < j$, the input graph $G$ of \pname{Number List Coloring} has the groups of gadgets $\mathcal{G}(i,j)$ and $\mathcal{G}(j,i)$, each containing $M$ edge gadgets corresponding to the edges of $E(V_i, V_j)$.
We say that $\mathcal{G}(i,j)$ is the \textit{forward group} and that $\mathcal{G}(j,i)$ is the \textit{backward group}.
For the description of the gadgets, we always assume $i < j$, $e \in E(V_i, V_j)$ with $u \in V_i$ and $v \in V_j$.

\medskip\noindent \textbf{Forward Edge Gadget.}
The gadget $G(i,j,e)$ has a root vertex $r(i, j, e)$, with list $\{\sigma(i,j), \sigma'(i, j)\}$, and two neighbors, both with the list $\{\sigma(i,j), y(i,j)\}$, which for convenience we call $a(i,j,e)$ and $b(i,j,e)$.
We equate membership of edge $e$ in the solution to \pname{Multicolored Clique} to the coloring of $r(i,j,e)$ with $\sigma(i,j)$.
When discussing the vertices of the remaining vertices of the gadget, we say that a vertex is \textit{even} if its distance to $r(i,j,e)$ is even, otherwise it is \textit{odd}.
To $a(i,j,e)$, we append a path with $2e_\down + 2(k-1)u_\down$ vertices.
First, we choose $e_\down$ even vertices to assign the list $\{y(i, j), \varepsilon'(i, j)\}$.
Next, for each $r$ in $j < r \leq k$, choose $u_\down$ even vertices to assign the list $\{y(i,j), \tau'_i(j,r)\}$.
Similarly, for each $s \neq i$ satisfying $s < j$, choose $u_\down$ even vertices and assign the list $\{y(i,j), \tau_i(s,j)\}$.
All the odd vertices - except $a(i,j,e)$ and $b(i,j,e)$ - are assigned the list $\{y(i,j), x(i,j)\}$.
The path appended to $b(i,j,e)$ is similarly defined, except for two points: (i) the length and number of chosen vertices are proportional to $e_\up$ and $u_\up$; and (ii) when color $\varepsilon(i,j)$ (resp. $\tau_i(s,r)$) should be in the list, we add $\varepsilon'(i,j)$ (resp. $\tau'_i(s,r)$), and vice-versa.
For an example of the edge gadget, please refer to Figure~\ref{fig:for_edge_gadget}.

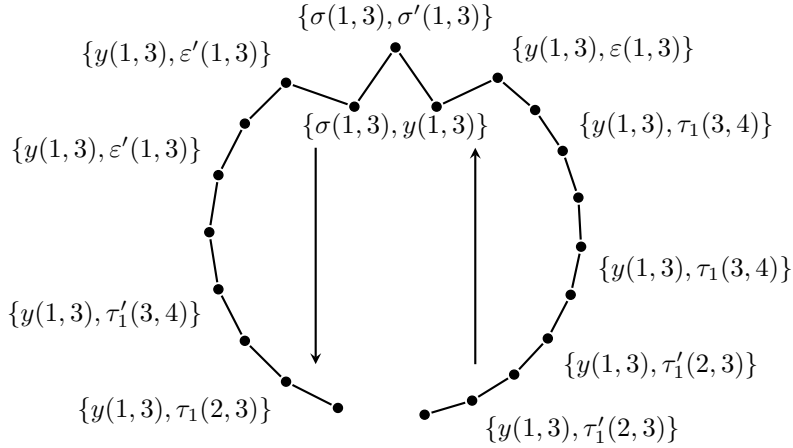
\begin{figure}[!htb]
    \centering
     \begin{tikzpicture}[scale=0.7]
         \GraphInit[unit=3,vstyle=Normal]
         \SetVertexNormal[Shape=circle, FillColor=black, MinSize=3pt]
         \tikzset{VertexStyle/.append style = {inner sep = \inners, outer sep = \outers}}
         \SetVertexLabelOut
         \Vertex[a=90, d=3.5, Lpos=90, Math, L={\{\sigma(1,3), \sigma'(1,3)\}}]{r}
         \Vertex[a=72, d=2.5, NoLabel]{b}
         \node at (0, 2) {$\{\sigma(1,3), y(1,3)\}$};
         \foreach \i in {0,1,2,3,4} {
            \pgfmathsetmacro{\ae}{72 - 15.3*(2*\i+1)}
            \pgfmathsetmacro{\ao}{\ae -15.3}
            \Vertex[a=\ao, d=3.5, NoLabel]{ou\i}
         }
         \foreach \i in {0} {
            \pgfmathsetmacro{\ae}{72 - 15.3*(2*\i+1)}
            \newcommand{\lab}{\varepsilon(1,3)}
            \Vertex[a=\ae, d=3.5, Lpos=\ae, Math, L={\{y(1,3), \lab\}}]{eu\i}
            \Edge(eu\i)(ou\i)
         }
         \foreach \i in {1,2} {
            \pgfmathsetmacro{\ae}{72 - 15.3*(2*\i+1)}
            \newcommand{\lab}{\tau_1(3,4)}
            \Vertex[a=\ae, d=3.5, Lpos=\ae, Math, L={\{y(1,3), \lab\}}]{eu\i}
            \Edge(eu\i)(ou\i)
         }
         \foreach \i in {3,4} {
            \pgfmathsetmacro{\ae}{72 - 15.3*(2*\i+1)}
            \newcommand{\lab}{\tau'_1(2,3)}
            \Vertex[a=\ae, d=3.5, Lpos=\ae, Math, L={\{y(1,3), \lab\}}]{eu\i}
            \Edge(eu\i)(ou\i)
         }
         
         \foreach \i in {0,1,2,3} {
            \pgfmathtruncatemacro{\n}{\i+1}
            \Edge(eu\n)(ou\i)
         }

         \Vertex[a=108, d=2.5, NoLabel]{a}
         \foreach \i in {0,1,2,3} {
            \pgfmathsetmacro{\ae}{108 + 18*(2*\i+1)}
            \pgfmathsetmacro{\ao}{\ae + 18}
            \Vertex[a=\ao, d=3.5, NoLabel]{od\i}
         }
         \foreach \i in {0,1} {
            \pgfmathsetmacro{\ae}{108 + 18*(2*\i+1)}
            \newcommand{\lab}{\varepsilon'(1,3)}
            \Vertex[a=\ae, d=3.5, Lpos=\ae, Math, L={\{y(1,3), \lab\}}]{ed\i}
            \Edge(ed\i)(od\i)
         }
         \foreach \i in {2} {
            \pgfmathsetmacro{\ae}{108 + 18*(2*\i+1)}
            \newcommand{\lab}{\tau'_1(3,4)}
            \Vertex[a=\ae, d=3.5, Lpos=\ae, Math, L={\{y(1,3), \lab\}}]{ed\i}
            \Edge(ed\i)(od\i)
         }
         \foreach \i in {3} {
            \pgfmathsetmacro{\ae}{108 + 18*(2*\i+1)}
            \newcommand{\lab}{\tau_1(2,3)}
            \Vertex[a=\ae, d=3.5, Lpos=\ae, Math, L={\{y(1,3), \lab\}}]{ed\i}
            \Edge(ed\i)(od\i)
         }
         
         \foreach \i in {0,1,2} {
            \pgfmathtruncatemacro{\n}{\i+1}
            \Edge(ed\n)(od\i)
         }
         \Edges(ed0, a, r, b, eu0)
         \draw[->,thick,>=stealth] (-1.5, 1.6) -- (-1.5,-2.5);
         \draw[->,thick,>=stealth] (1.5,-2.5) -- (1.5, 1.6);
     \end{tikzpicture}
    \caption{Example of a forward edge gadget $G(1,3,e)$ of group $\mathcal{G}(1,3)$, with $k = 4$, $Z = 3$, $e_\down = 2$, and $u_\down = 1$. Vertices with no explicit list have list equal to $\{y(1,3), x(1,3)\}$.}
    \label{fig:for_edge_gadget}
\end{figure}

\medskip\noindent \textbf{Backward Edge Gadget.}
Gadget $G(j,i,e)$ has vertices $r(j, i, e)$, $a(j,i,e)$, and $b(j,i,e)$ defined similarly as to the forward gadget, with the root vertex having the list $\{\sigma(j,i), \sigma'(j,i)\}$, while the other two have the list $\{\sigma(j,i), y(j,i)\}$.
To $a(j,i,e)$, we append a path with $2e_\down + 2(k-1)v_\down$ vertices.
First, choose $e_\down$ even vertices to assign the list $\{y(j,i), \varepsilon(i, j)\}$. 
Now, for each $r$ in $j < r \leq k$, choose $v_\down$ even vertices to assign the list $\{y(j,i), \tau'_j(i,r)\}$.
Then, for each $s \neq i$ satisfying $s < j$, choose $v_\down$ even vertices and assign the list $\{y(j,i), \tau_j(s,i)\}$.
All the odd vertices are assigned the list $\{y(j,i), x(j,i)\}$.
The path appended to $b(j,i,e)$ is similarly defined, except that: (i) the length and number of chosen vertices are proportional to $e_\up$ and $v_\up$; and (ii) when the color $\varepsilon(i,j)$ (resp. $\tau_j(s,r)$) is in the list, we replace it with $\varepsilon'(i,j)$ (resp. $\tau'_j(s,r)$), and vice-versa.
Note that, for every edge gadget, either forward or backward, the number of vertices is equal to $3 + 2(e_\up + e_\down) + 2(k-1)(u_\up + u_\down) = 3 + 2kZ$. 
We say that $G(i,j,e)$ is \textit{selected} if $r(i,j,e)$ is colored with $\sigma(i,j)$, otherwise it is \textit{passed}.

\medskip\noindent \textbf{Numerical Targets.}
Before defining the numerical targets, given by $h(\cdot)$, recall that $|E(V_i, V_j)| = M$ for every pair $i,j$ and that, for every vertex $u$ and edge $e$, the identification numbers satisfy the identity $v_\up + v_\down = Z$ and $e_\up + e_\down = Z$.
We present the numerical targets of our instance - and some intuition - below.

\begin{itemize}
    \item[Selection:] $h(\sigma(i,j)) = 1 + 2(M-1)$ and $h(\sigma'(i,j)) = M-1$.
    Since only one edge may be chosen from $V_i$ to $V_j$, the non-selection color $\sigma'(i,j)$ must be used in $M-1$ edges of $\mathcal{G}(i,j)$.
    Thus, exactly one $G(i,j,e)$ is selected and, to achieve the target of $1 + 2(M-1)$, for every $f \in E(V_i, V_j) \setminus \{e\}$, both $a(i,j,f)$ and $b(i,j,f)$ must also be colored with $\sigma(i,j)$.
    \item[Helper:] $h(y(i,j)) = 2 + kMZ$ and $h(x(i,j)) = kMZ - kZ$.
    The goal here is that, if $G(i,j,e)$ is selected, all the odd positions must be colored with $y(i,j)$, otherwise every even position must be colored with it.
    In the latter case, the odd positions of all but one gadget of $\mathcal{G}(i,j)$ must be colored with $x(i,j)$.
    \item[Symmetry:] $h(\varepsilon(i,j)) = h(\varepsilon'(i,j)) = Z$.
    If the previous condition holds and $r(i,j,e)$ is colored with $\sigma(i,j)$, then $\varepsilon(i,j)$ appears in $e_\up$ vertices of the gadget rooted at $r(i,j,e)$.
    To meet the target $Z$, $e_\down$ vertices of another gadget must also be colored with it, as we show, the only way is if $r(j,i,e)$ is colored with $\sigma(j,i)$.
    \item[Consistency:] $h(\tau_i(s,r)) = h(\tau'_i(s,r)) = Z$.
    Similar to symmetry colors.
\end{itemize}

\begin{lemma}
    \label{lem:forward_nlc}
    If $H$ has a $k$-multicolored clique, then $G$ admits a list coloring meeting the numerical targets.
\end{lemma}

\begin{proof}
    Let $Q$ be a clique of $H$ of size $k$.
    Now, for each $i,j \in [k]$, with $i \neq j$, and $e \in E(V_i, V_j)$, we color $G$ as follows: if $e \in Q$, color $r(i,j,e)$ with $\sigma(i,j)$, color every odd vertex with $y(i,j)$ and every even vertex with the unique available color to it; otherwise, color $r(i,j,e)$ with $\sigma'(i,j)$, color every even vertex with $y(i,j)$ and all odd vertices with the unique available color.
    This concludes the construction of the coloring $\varphi$ of $G$.
    
    As to the numerical targets, note that the colors of $\mathcal{S}$ and $\mathcal{S}'$ are used the appropriate number of times since only one edge of $E(V_i, V_j)$ belongs to $Q$.
    For each color $y(i,j) \in \mathcal{Y}$ and edge $e = uv$ with $u \in V_i$, gadget $G(i,j,e)$ has, at least, $e_\up + e_\down + (k-1)(u_\up + u_\down) = kZ$ vertices colored with $y(i,j)$, since either all odd vertices or even vertices are colored with it.
    For the remaining two uses of $y(i,j)$, note that the selected gadget $G(i,j,e)$ has $a(i,j,e)$ and $b(i,j,e)$ also colored with $y(i,j)$.
    As to the other helper color, $x(i,j)$, we use it only in passed gadgets and, in this case, in every odd vertex (except the $a$'s and $b$'s); this sums up to $\sum_{e \in E(V_i, V_j) \setminus Q} e_\up + e_\down + (k-1)(u_\up + u_\down) = kMZ - kZ$.
    In terms of symmetry colors, $\varphi$ only uses $\varepsilon(i,j)$ on the selected gadgets $G(i,j,e)$ and $G(j,i,e)$ ($i < j$), in particular, $\varepsilon(i,j)$ is used $e_\down$ times in $G(j,i,e)$ and $e_\up$ times on $G(i,j,e)$, so it holds that $|\varphi_{\varepsilon(i,j)}| = e_\up + e_\down = Z$.
    Note that the same argument applies for color $\varepsilon'(i,j)$.
    Finally, for consistency colors, note that $\tau_i(r,s)$ is used only in the selected gadgets $G(i,r, e)$ and $G(i,s,f)$, specifically, if $i < r < s$, $\tau_i(r,s)$ is used $u_\up$ times in $G(i,r,e)$ and $u_\down$ times in $G(i,s,f)$, since edges $e,f$ must be incident to the same vertex of $Q \cap V_i$.
    Consequently $\tau_i(r,s)$ is used in $u_\up + u_\down = Z$ vertices.
    The reasoning for $\tau'_i(r,s)$ is similar.
\end{proof}

We now proceed to the proof of the converse.
Lemma~\ref{lem:parity} guarantees that the decision made at the root of a gadget propagates throughout the entire structure; Lemma~\ref{lem:symmetry} ensures that the edge selected from $V_j$ to $V_i$ is the same as the edge selected from $V_i$ to $V_j$; finally, Lemma~\ref{lem:consistency} equates the vertex of $V_i$ incident to the edge between $V_i$ and $V_j$ to the vertex incident to the edge between $V_i$ and $V_s$.

\begin{lemma}
    \label{lem:parity}
    In every list coloring of $G$ satisfying $h$, exactly one gadget of each $\mathcal{G}(i,j)$ is selected, each passed $G(i,j,e)$ has all of its $kZ$ even vertices colored with $y(i,j)$, and the selected $G(i,j,f)$ has all of its $2 + kZ$ odd vertices colored with $y(i,j)$.
\end{lemma}

\begin{proof}
    For the first statement, if no gadget was selected, $\sigma'(i,j)$ would be used $M > M - 1$ times; if more than one is selected, $\sigma'(i,j)$ does not meet the target.
    If $G(i,j,f)$ is selected then $a(i, j, f)$ and $b(i, j, f)$ are colored with $y(i, j)$. After removing these vertices, we are left with two even paths of which at most half of its vertices are colored with the same color.
    As such there are at most $2 + f_\down + (k-1)u_\down + f_\up + (k-1)u_\up = 2 + kZ$ vertices colored with $y(i, j)$ in this gadget; moreover this bound is achieved if and only if the odd vertices are colored with $y(i,j)$.
    For each passed $G(i,j,e)$, $a(i, j, e)$ and $b(i, j, e)$ are colored with $\sigma(i, j)$, otherwise the numerical target of $\sigma(i, j)$ cannot be met.
    After the removal of $a,b,r(i, j, e)$ we are again left with two even paths.
    Thus, at most $e_\down + (k-1)u_\down + e_\up + (k-1)u_\up = kZ$ vertices have color $y(i, j)$ in this gadget.
    To see that $y(i,j)$ must be used only for even vertices of $G(i,j,e)$ to meet this bound, note that, if this is not done, then $x(i,j)$ will never meet its bound, since exactly $kMZ - kZ$ vertices remain (after the coloring of $G(i,j,f)$) that can be colored with $x(i,j)$, which is precisely its target.
\end{proof}

\begin{lemma}
    \label{lem:symmetry}
    In every list coloring $\varphi$ of $G$, if $G(i,j,e)$ is selected, so is $G(j,i,e)$.
\end{lemma}
    
\begin{proof}
    Suppose $i < j$.
    By Lemma~\ref{lem:parity} we know that for a selected gadget every odd vertex is colored with $y(i, j)$, so each even vertex is colored with a non-helper color.
    Note that color $\varepsilon'(i,j)$ is used $e_\down$ times in gadget $G(i,j,e)$.
    Now, we need to select one backward gadget of $\mathcal{G}(j,i)$; suppose we select gadget $G(j,i,f)$, $f \neq e$.
    Again by Lemma~\ref{lem:parity}, the number of occurrences of $\varepsilon'(i,j)$ is $f_\up$ times in $G(j,i,f)$, but we have that $e_\down + f_\up \neq Z$, a contradiction that $\varphi$ satisfies $h$.
\end{proof}

\begin{lemma}
    \label{lem:consistency}
    In every list coloring $\varphi$ of $G$, if $G(i,j,e)$ is selected and $e = uv$, then, for every $s \neq i$, the edge $f$ of $H$ corresponding to the selected gadget $G(i, s, f)$ must be incident to $u$.
\end{lemma}

\begin{proof}
    We divide our proof in two cases.
    First, suppose $i < j$ and $s < j$.
    By Lemma~\ref{lem:parity} $\tau_i(s,j)$ is used in $u_\up$ vertices of $G(i,j,e)$.
    Now, note that the only possible gadgets $G(i,s,f)$ that can be chosen such that the endpoint $w$ of $f$ in $V_i$
    satisfies $u_\up + w_\down = Z$ must have $w = u$, since $\tau_i(s,j)$ is used $w_\down$ times in gadget $G(i,s,f)$.
    The case $j < s$ is similar, but we replace $\tau_i$ with $\tau'_i$.
\end{proof}

\begin{theorem}
    \label{thm:nlc_hard}
    \pname{Number List Coloring} on paths parameterized by the number of colors that appear on the lists is $\W[1]$-$\Hard$.
\end{theorem}

\begin{corollary}
    \pname{Equitable List Coloring} on paths parameterized by the number of colors that appear on the lists is $\W[1]$-$\Hard$.
\end{corollary}

\begin{corollary}
    \pname{Equitable Coloring} parameterized by the number of colors and distance to disjoint paths is $\W[1]$-$\Hard$.
\end{corollary}

Theorem~\ref{thm:nlc_hard} and its corollaries follow from the previous lemmas.
For the first corollary, we add isolated vertices with lists of size one to the input of \pname{Number List Coloring} so as to make all colors have the same numerical target.
For the second, we add a clique of size $r$, the number of colors of the instance of \pname{Equitable List Coloring}, and label them using the integers $[r]$; afterwards, for each vertex $u$ of the input graph that does not have color $i$ in its list, we add an edge between the $i$-th vertex of the clique and $u$.

    \section{Distance to disjoint paths of bounded size}
\label{sec:bpath}

The reduction in Section~\ref{sec:dpath} heavily relies on the length of the disjoint paths to show that \pname{Equitable Coloring} is \WH\ parameterized by distance to disjoint paths and number of colors.
In this section, we show that, by replacing the number of colors with the length of the longest path in the parameterization, the problem becomes \FPT.
We divide our analysis in two cases: when the number $k$ of colors is at most $|U| + 2$ and when it is bigger.

\subsection{Bounded number of colors}
\label{sec:bound_nlc}

We first study the case where the number of colors is bounded by $|U| + 2$.
To this end, we show how to translate an instance $(G, L, h)$ of \pname{Number List Coloring} where the input graph is a collection of disjoint paths into an instance of \pname{Integer Linear Programming}, which is fixed-parameter tractable when parameterized by the number of variables~\cite{ilp}.
If $A = \angled{a_1, \dots, a_\ell}$ is one of the disjoint paths of $G$, we say that the \textit{pattern} of $A$, denoted by $\pat(A)$, is the sequence $\angled{L(a_1), \dots, L(a_\ell)}$.
A set of disjoint paths $\mathcal{A}$ is \textit{compatible} if, for every two $A,B \in \mathcal{A}$, $\pat(A) = \pat(B)$, and there is no other path of $G$ with the same pattern.
In an abuse of notation, we say that $\pat(A) = \pat(\mathcal{A})$.
We denote the partition of $G$ into compatible sets by $\mathcal{P}(G)$ and say that a pattern $p$ is in $\mathcal{P}(G)$ if there is some $\mathcal{A} \in \mathcal{P}(G)$ with $\pat(\mathcal{A}) = p$.

\begin{lemma}
    \label{lem:bound_nlc}
    If $(G, L, h)$ is an instance of \pname{Number List Coloring} such that each connected component of $G$ is a path of length at most $\ell$ and $k = |\bigcup_{v \in V(G)} L(v)|$, then there is an algorithm that runs in $f(|\mathcal{P}(G)|, k, \ell)n^\bigO{1}$ time and correctly solves $(G, L, h)$.
\end{lemma}

\begin{proof}
    Our goal is to construct an integer linear program that correctly computes a coloring of $G$ subject to the coloring constraints imposed by $L$ that meets the numerical targets given by $h$.
    For each $\mathcal{A}_p \in \mathcal{P}(G)$, we write $\pat(\mathcal{A}_p) = \{L_{p,1}, \dots, L_{p,q_p}\}$, where $q_p = |\pat(\mathcal{A}_p)|$, and construct a flow network $H_p$ as follows: $V(H_p) = \{s_p, t_p\} \bigcup_{j \in [q_p], i \in L_{p,j}} \{u_{p,j,i}, v_{p,j,i}\}$. The source $s_p$ has an outbound arc to each $u_{p,1,i}$, every $u_{p,j,i}$ is an in-neighbor of $v_{p,j,i}$, every $v_{p,j,i}$ has an outbound arc to each $u_{p,j+1,a}$ that satisfies $i \neq a$, and each $v_{p,q_p,i}$ is an in-neighbor of the sink $t_p$; all arcs have capacity equal to $|\mathcal{A}_p|$.
    An example of $H_p$ is shown in Figure~\ref{fig:bpath}.

    \begin{figure}[!htb]
        \centering
            \begin{tikzpicture}[scale=0.7]
                \GraphInit[unit=3,vstyle=Normal]
                \SetVertexNormal[Shape=circle, FillColor=black, MinSize=3pt]
                \tikzset{VertexStyle/.append style = {inner sep = \inners, outer sep = \outers}}
                \SetVertexLabelOut
                \Vertex[x=-1, y=0, Math, Lpos=180]{s}
                
                \Vertex[x=1, y=1.5, Math, Lpos=90, L={u_{1,1}}]{u11}
                \Vertex[x=3, y=1.5, Math, Lpos=90, L={v_{1,1}}]{v11}
                
                \Vertex[x=1, y=0, Math, Lpos=90, L={u_{1,2}}]{u12}
                \Vertex[x=3, y=0, Math, Lpos=90, L={v_{1,2}}]{v12}

                \Vertex[x=6, y=1.5, Math, Lpos=90, L={u_{2,1}}]{u21}
                \Vertex[x=8, y=1.5, Math, Lpos=90, L={v_{2,1}}]{v21}
                
                \Vertex[x=6, y=-1.5, Math, Lpos=270, L={u_{2,3}}]{u23}
                \Vertex[x=8, y=-1.5, Math, Lpos=270, L={v_{2,3}}]{v23}

                \Vertex[x=11, y=0, Math, Lpos=90, L={u_{3,2}}]{u32}
                \Vertex[x=13, y=0, Math, Lpos=90, L={v_{3,2}}]{v32}
                
                \Vertex[x=11, y=-1.5, Math, Lpos=270, L={u_{3,3}}]{u33}
                \Vertex[x=13, y=-1.5, Math, Lpos=270, L={v_{3,3}}]{v33}
                
                \Vertex[x=15, y=0, Math, Lpos=0]{t}
                \Edges[style={->}](s, u11, v11, u23, v23, u32, v32, t)
                \Edges[style={->}](s, u12, v12, u21, v21, u33, v33, t)
                \Edges[style={->, bend right}](v12, u23)
                \Edges[style={->}](v21, u32)
        \end{tikzpicture}
    \caption{Flow network $H_p$ where $\pat(\mathcal{A}_p) = \{\{1,2\}, \{1,3\}, \{2, 3\}\}$ and all arcs have capacity equal to $|\mathcal{A}_p|$. We omit $p$ from the subscript.\label{fig:bpath}}
    \end{figure}
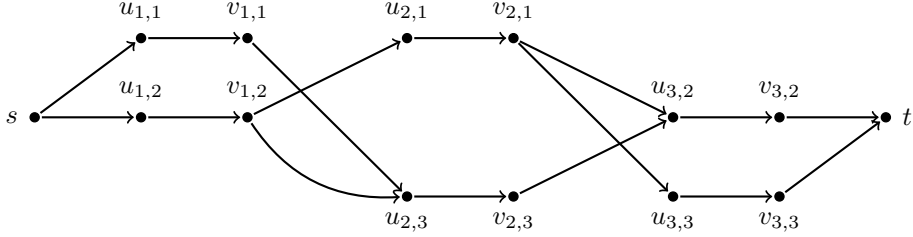
    Now, let $\Pi(H_p, \mathcal{A}_p, \eta)$ be the following set of linear constraints, where $x(e)$ is the variable storing how much flow is passing through arc $e \in E(H_p)$, $E_{p,i} = \bigcup_{j \in [q_p]} \{(u_{p,j,i}, v_{p,j,i})\}$, $E^-(v)$  is the set of inbound arcs of $v$, $E^+(v)$ the set of outbound arcs, and $\eta(p,i)$ is (for now) an integer limiting how many vertices of $\mathcal{A}_p$ must be colored with color $i$:
    
    \begin{align}
        \sum_{e \in E^-(v)} x(e) - \sum_{e \in E^+(v)} x(e) &=\   0                &\forall v \in V(H) \setminus \{s_p, t_p\} \label{eq:flow_conserve}\\
        \sum_{e \in E^+(s)} x(e) = \sum_{e \in E^-(t)} x(e) &=\   |\mathcal{A}_p|  & \label{eq:source_sink}\\
        \sum_{e \in E_{p,i}} x(e)                           &=\   \eta(p,i)        &\forall i \in \bigcup_{j \in [q_p]} L_{p,j} \label{eq:eta_p}\\
        0 \leq x(e)                                                &\leq |\mathcal{A}_p|  &\forall e \in E(H_p) \label{eq:capacity}
    \end{align}
    
    Before proving the next claim, we note that Equations~\ref{eq:flow_conserve},~\ref{eq:source_sink}, and~\ref{eq:capacity} are standard flow conservation constraints, while Equation~\ref{eq:eta_p} imposes the constraints on how many times the flow must pass through arcs of $E_{p,i}$.
    Intuitively, each flow unit corresponds to one path $A \in \mathcal{A}_p$ and the $(s,t)$-path taken by the unit in $H_p$ corresponds to a coloring of $A$ in $G$.
    
    \begin{claim}
        $\mathcal{A}_p$ has a coloring respecting the lists given by $L$ and the numerical targets given by $\eta(p,i)$ if and only if $\Pi(H_p, \mathcal{A}_p, \eta)$ is feasible.
    \end{claim}
    
    \begin{cproof}
        For the forward direction, let $\varphi$ be a coloring of $H_p$ that respects $L$ and meets $\eta(p,i)$.
        For each $A_r \in \mathcal{A}_p$, we push one unit of flow from $s_p$ along the path of $H_p$ described by the coloring of $A_r = \{v_1, \dots, v_{q_p}\}$, so that  $(u_{p,j,i}, v_{p,j,i})$ is in the path if $v_j \in \varphi_i$, and all other arcs are chosen appropriately; this procedure does not violate Equations~\ref{eq:flow_conserve} through~\ref{eq:capacity} since this is an $(s,t)$-flow, each arc is picked at most once per path of $\mathcal{A}_p$ and the number of flow units coursing through the arcs of each $E_{p,i}$ is equal to $\eta(p,i)$ since $|\varphi_i| = \eta(p,i)$.
        
        For the converse, suppose that $\mathbf{x}$ is an assignment of the variables satisfying $\Pi(H_p, \mathcal{A}_p, \eta)$.
        Our proof is by induction on $|\mathcal{A}_p|$.
        For the base case, where $A = \angled{a_1, \dots, a_{q_p}}$ is the unique element of $\mathcal{A}_p$, we color $a_j$ with $i$ if $x((u_{p,j,i}, v_{p,j,i}))$; since $(v_{p,j,i}, u_{p, j+1, i}) \notin E(H_q)$, no two adjacent vertices of $A$ have the same color and, since Equation~\ref{eq:eta_p} is satisfied and $x(e) \leq 1$ for very edge of $E(H_p)$, $|\varphi_i| = \eta(p,i)$ for every color $i$.
        For the general case, let $W = \angled{s, u_1, v_1, \dots, u_{q_p}, v_{q_p}, t}$ be a path of $H_p$ so that every arc $e$ with both endpoints in $W$ has $x(e) \neq 0$ and the set of all saturated arcs of $H_p$ (possibly empty) is contained in $W$; there exists such a path since $H_p$ is acyclic and $\sum_{i \in [k]} x((u_{p,j,i}, v_{p,j,i})) = |\mathcal{A}_p|$.
        Now, let $\mathbf{x}(W)$ be the vector where the positions corresponding to each arc in $W$ are equal to one, $\eta' = \{\eta(p,i) - |E(W) E_{p,i}| \cap \mid i \in [q_p]\}$ be a new set of numerical targets, and $\mathcal{A}'_p$ be a subset of size $|\mathcal{A}_p| - 1$ of $\mathcal{A}_p$.
        Note that $\mathbf{x}' = \mathbf{x} - \mathbf{x}(W)$ is a solution to $\Pi(H_p, \mathcal{A}'_p, \eta')$ since $\max_e{\mathbf{x}(e)} \leq |\mathcal{A}_p| - 1$.
        By the inductive hypothesis, we can convert $\mathbf{x} - \mathbf{x}(W)$ into a coloring of $|\mathcal{A}_p| - 1$ paths with pattern $\pat(\mathcal{A}_p)$.
        For the $|\mathcal{A}_p|$-th path, we color it according to $W$ as done in the base case; such a coloring does not violate the capacity of the arcs since $\mathbf{x}'(e) \leq |\mathcal{A}_p - 1$ for each arc $e$ of $W$.
    \end{cproof}
    
    Armed with the claim, we may now formulate and prove correctness of the complete linear program $\Lambda(G, k, L, h)$.
    Note that, unlike in program $\Pi$, the $\eta(p,i)$'s are variables of program $\Lambda$.

    \begin{align}
        \Pi(H_p, \mathcal{A}_p, \eta)              &                     &\forall p \in \mathcal{P}(G) \label{eq:pies}\\
        \sum_{i \in [k]} \eta(p,i)            &= |\mathcal{A}_p|q_p &\forall p \in \mathcal{P}(G) \label{eq:totality_vert}\\
        \sum_{p \in \mathcal{P}(G)} \eta(p,i) &= h(i)               &\forall i \in [k] \label{eq:totality_h}\\
        0 \leq \eta(p,i)                      &\leq h(i)            &\forall p \in \mathcal{P}(G), i \in [k]
    \end{align}
    
    To see that $\Lambda(G, k, L, h)$ is feasible if and only if $(G, L, h)$ is a positive instance of \pname{Number List Coloring}, note that Equations~\ref{eq:totality_vert} and~\ref{eq:totality_h} are satisfied if and only if we allow the algorithm to use one color for each vertex in $\mathcal{A}_p$ and each color precisely the number of times demanded by $h$, respectively.
    Moreover, by our previous Claim, if every $\Pi(H_p, \mathcal{A}_p, \eta)$ is feasible, then there is a coloring of $\bigcup_{p \in \mathcal{P}(G)} G[\mathcal{A}_p] = G$ satisfying the desired numerical targets.
    As to the running time of the algorithm, $\Lambda(G, k, L, h)$ has $\bigO{k\ell|\mathcal{P}(G)|}$ variables; since each connected component of $G$ has at most $\ell$ vertices and to each vertex of $G$ one of $2^k$ possible lists was assigned, we have that $G$ has at most $2^{k\ell}$ compatible sets, i.e $|\mathcal{P}(G)| \leq 2^{k\ell}$.
    As shown in~\cite{ilp}, \pname{Linear Integer Programming} can be solved in \FPT\ time on the number of variables, so it follows that our algorithm runs in $f(|\mathcal{P}(G)|, k, \ell)n^\bigO{1}$ time.
\end{proof}

To solve \pname{Equitable Coloring} itself, note that, for each of the $|U|^{|U|}$ possible colorings of the $P_\ell$-modulator $U \subseteq V(G)$, we have that each vertex in $G \setminus U$ has a list of available colors of size at most $|U| + 2$.
Fixed this precoloring $\varphi$ of $U$ and which of the $k$ colors will be used $\ceil{\frac{n}{k}}$ times, by Lemma~\ref{lem:bound_nlc}, we can decide in $f(|\mathcal{P}_\varphi(G \setminus U)|, |U| + 2, \ell)n^\bigO{1}$ time whether $G \setminus U$ can be colored respecting the lists and meeting the targets imposed by the precoloring $\varphi$

\subsection{Unbounded number of colors}

The second case of our analysis has some similar ideas to the previous one.
This time, however, we are going to extend the precoloring of $U$ so that every color that appears in $\varphi$ is taken care of using an integer linear program, while the remaining $k - |U| \geq 3$ colors will be assigned using the Hajnal-Szemerédi theorem~\cite{hajnal_szmeredi_theorem}, which states that every graph can be equitably $k$-colored if $k$ is at least the maximum degree plus one.

\begin{lemma}
    \label{lem:unbound_eq}
    Let $(G,k)$ be an instance of \pname{Equitable Coloring}, $U \subset V(G)$ be a $P_\ell$-modulator of $G$, $\varphi$ a precoloring of $U$, and $\gamma: \varphi(U) \mapsto \{\floor{\frac{n}{k}}, \ceil{\frac{n}{k}}\}$ be a function dictating how many times each color in $\varphi(U)$ must be used in the entire graph.
    There exists an algorithm that runs in $f(|U|, \ell)n^\bigO{1}$ time and decides if $\varphi$ can be extended to an equitable $k$-coloring $G$ that respects $\gamma$.
\end{lemma}

\begin{proof}
    Throughout this proof, we use $|\varphi(U)| = \rho$.
    Let $(G \setminus U, L, h)$ be the instance of \pname{Number List Coloring} that, for each $v \in V(G) \setminus U$, has $L(v) = [\rho + 2] \setminus \varphi(N_G(v) \cap U)$, and $h$ be the function $h: [\rho+2] \mapsto [n]$ where $h(i) = \gamma(i) - |\varphi_i(U)|$ for every $i \in \varphi(U)$, $h(\rho+1) = \floor{\frac{n - \sum_{i \in \varphi(U)}\gamma(i)}{2}}$, and $h(\rho+2) = \ceil{\frac{n - \sum_{i \in \varphi(U)}\gamma(i)}{2}}$.
    
    We claim that $(G \setminus U, L, h)$ is a positive instance if and only if we can extend $\varphi$ to an equitable $k$-coloring of $G$ that respects $\gamma$.
    Let $\varphi'$ be the desired extension and $X$ the set of vertices not colored with a color of $\varphi(U)$.
    We know that $G[X]$ is a family of disjoint paths.
    We prove that we can greedily recolor $X$ with two colors $\rho+1$ and $\rho+2$ in order to meet the targets given by $h$.
    Since we are going to color paths with only two colors, vertices on even positions are colored with one color, and in odd positions with the other.
    As such, paths of even length do not increase the difference between the number of vertices colored with $\rho + 1$ and $\rho + 2$, while paths of odd length increase it by one unit.
    As such, we begin by recoloring all even length paths properly but arbitrarily.
    Before proceeding, note that $h(\rho+1) = h(\rho+2)$ if and only if we have an even number of odd length paths.
    If $\{B_1, \dots, B_r\}$ are all such paths, we color the odd-indexed paths starting with $\rho+2$ and the even-indexed paths starting with $\rho+1$; this guarantees that each of the was used at least $\sum_{i \in [r]} \floor{\frac{|B_i|}{2}} = \floor{\frac{n - \sum_{i \in \varphi(U)}h_i}{2}}$ times, with $\rho+2$ having an additional usage if and only if $G[X]$ has an odd number of odd length paths.
    Since $\rho+1$ and $\rho+2$ are not used in any other vertex, we have obtained a proper coloring of $G \setminus U$ that respects $L$ and satisfies the numerical targets.
    
    For the converse, let $\mathbf{x}$ be an assignment satisfying the integer linear program $\Lambda(G \setminus U, \rho + 2, L, h)$.
    By Lemma~\ref{lem:bound_nlc}, $\mathbf{x}$ has an associated coloring of the vertices of $G \setminus U$ meeting the targets given by $h$.
    It follows immediately that all colors $i \in \varphi(U)$ were used $h(i) + |\varphi_i(U)| = \gamma(i)$ times.
    Now, let $X \subseteq V(G) \setminus)$ be the set of vertices colored with $\rho+1$ or $\rho+2$, $k_1$ (resp. $k_2$) be the number of colors of an equitable $k$-coloring of $G$ that must have size $\ceil{\frac{n}{k}}$ (resp. $\floor{\frac{n}{k}}$), $\rho_1 = |\{i \mid \gamma(i) = \ceil{\frac{n}{k}}\}|$, and $\rho_2 = \rho - \rho_1$. 
    We note that $|X| = n - \rho_1\ceil{\frac{n}{k}} - \rho_2\floor{\frac{n}{k}} = (k_1 - \rho_1) \ceil{\frac{n}{k}} + (k_2 - \rho_2) \floor{\frac{n}{k}}$, i.e if we want to equitably $(k - \rho)$-color $X$, each color must use $\ceil{\frac{n}{k}}$ or $\floor{\frac{n}{k}}$ colors.
    Since $G[X]$ has maximum degree two and $k - \rho \geq 3$, the Hajnal-Szemerédi theorem guarantees that there is an equitable $(k - \rho)$-coloring of $G[X]$, which can be computed in polynomial time~\cite{fast_equitable}; since $G[X]$ uses no color in $\varphi(V(G) \setminus X)$, we have found an equitable $k$-coloring of $G$.
    The overall running time is the same as in Lemma~\ref{lem:bound_nlc}, since the additional step of transforming the coloring on $\rho+2$ colors of $G \setminus U$ can be performed in polynomial time.
\end{proof}

The next theorem follows directly from the previous two lemmas.
The only additional steps necessary are guessing a coloring of the $P_\ell$-modulator and which of these colors of the modulator will be used in $\ceil{\frac{n}{k}}$ vertices.
Moreover, we do not need to assume that the modulator has been given to us beforehand: we may adapt the well known vertex cover branching algorithm~\cite{cygan_parameterized} to find a $P_\ell$-modulator of size $t$ in $\ell^\bigO{t}n^2$ time.

\begin{theorem}
    \label{thm:unbound_eq}
    \pname{Equitable Coloring} can be solved in \FPT\ time when parameterized by distance to paths of length at most $\ell$ and $\ell$.
\end{theorem}
    \section{A linear kernel for distance to clique}

Let $(G, k)$ be the input to \pname{Equitable Coloring}, $U \subseteq V(G)$ be such that $G \setminus U$ is a complete graph, and define $n = |V(G)|$.
Our kernel is a direct consequence of the following lemma.

\begin{lemma}
    \label{lem:dclique}
    If $k > |U|$, then $(G, k)$ can be solved in polynomial time.
\end{lemma}

\begin{proof}
    Let $d = |U|$ and $Q = V(G) \setminus U$.
    Since $G[Q]$ is a clique, we know that $k - d > 0$ colors must necessarily be used only on $Q$ and only once, so no color may be used more than twice to obtain an equitable $k$-coloring of $G$.
    As such, let $M = \{u_1v_1, \dots, u_mv_m\}$ be a maximum matching of $\overline{G}$, i.e. a family of pairs of vertices of maximum size where $u_iv_i \notin E(G)$ for every $i \in [m]$, which can be found in polynomial time~\cite{edmonds_matching}.
    Note that $\{u_i,v_i\} \nsubseteq Q$, since $\overline{G}[Q]$ is an independent set.
    Moreover, since $M$ is of maximum size, $m$ is the largest number of colors that may be simultaneously used in two vertices, i.e. if $n - 2m > k - m$, we answer \NOi\ --- $G$ has no equitable $k$-coloring --- because the previous condition implies that at least one more color must be used in two vertices.
    For the final case, where $n \leq k + m$, let $k'$ be such that $n = k'\ceil{{\frac{n}{k}}} + (k - k')\floor{{\frac{n}{k}}}$ and note that $k' \leq m$, otherwise $G$ would not be equitably $k$-colorable.
    As such, for each $i \in [k']$, pick the non-edge $u_iv_i \in M$ and color both $u_i$ and $v_i$ with $i$, and color the remaining $n - 2k' = k - k'$ vertices with an arbitrary but unique color of the set $[k] \setminus [k']$; since $u_iv_i \notin E(G)$, this coloring is proper.
\end{proof}

\begin{theorem}
    When parameterized by the distance to clique $d$, \pname{Equitable Coloring} admits a kernel with $4d$ vertices. 
\end{theorem}

\begin{proof}
    Suppose we are not given $U$ in advance.
    Note that there is a naive polynomial time algorithm that yields a clique modulator $U$ of size at most $2d$: for each non-edge $uv$ of $G$, remove both $u$ and $v$ from $G$.
    As such, we obtain a modulator $U$ where $|U| \leq 2d$ and $G \setminus U$ is isomorphic to a complete graph.
    If $k > |U|$, we use Lemma~\ref{lem:dclique} and provide the answer for $(G, k)$ in polynomial time.
    Otherwise, $k \leq |U| \leq 2d$, and $|V(G) \setminus U|$ must be at most $k$, otherwise it is clearly impossible to even $k$-color $G \setminus U$.
    Consequently, $V(G) \leq k + |U| \leq 4d$.
\end{proof}
    \section{A cubic kernel for max leaf number}

We now show a polynomial kernel for \pname{Equitable Coloring} when parameterized by the max leaf number $\ell$.
As before, we always assume that $(G, k)$ is our input instance.
Our result heavily relies on a previous proof of Estivill-Castro et al.~\cite{max_subdivision} that $G$ has max leaf number $\ell$ if and only if it is a subdivision of a graph $H$ on $4\ell$ vertices, and a recent paper of Cappelle et al.~\cite{locating_lagos} that gives an upper bound of $5\ell + \floor{\frac{\ell}{2}}$ on the number of subdivided edges of $H$; graph $H$ is also called the \textit{host graph} of $G$.
We may assume that $k \geq 3$, otherwise the problem would be solvable in polynomial time by finding a maximum matching in the complement graph.
Throughout this section, let $H$ be the graph that $G$ is a subdivision of, let $P(u,v)$ be the path that replaced $uv \in E(H)$ in order to obtain $G$ --- note that $u,v \notin P(u,v)$ --- and let $\mathcal{P}_2(G)$ be the connected components of $G - H$.

Given a (partial) $k$-coloring $\varphi = \{\varphi_1, \dots, \varphi_k\}$ of $V(G)$, we assume that $|\varphi_i| \leq |\varphi_j|$ for every $i \leq j$; we say that $\varphi$ is \textit{$\alpha$-balanced} if $|\varphi_k| - |\varphi_1| = \alpha$.
Our first goal is to show how we can extend a partial  coloring $\varphi'$ of $G \setminus P$, for some $P \in \mathcal{P}_2(G)$, in such a way that the resulting coloring $\varphi$ is similarly balanced.

\begin{lemma}
    \label{lem:rebalance}
    Let $Q$ be an induced subgraph of $G$ that does not contain any vertex of some path $P \in \mathcal{P}_2(G)$ and $\varphi$ be an $\alpha$-balanced coloring of $Q$.
    If $|P| \geq k+2$, then it is possible to extend $\varphi$ into an $\alpha^*$-balanced coloring of $Q \cup P$, so that $\alpha^* \leq \alpha+1$ and $\alpha^* = \alpha+1$ only if $\alpha = 0$.
\end{lemma}

\begin{proof}
    The proof is by induction on $p = |P|$; also, let $s$ be the largest integer so that $|\varphi_1| = |\varphi_s|$ and $P = \angled{w_1, \dots, w_p}$.
    We may assume that, since each vertex of $P$ has two neighbors in $G$, the neighbors of the endpoints of $P$ are in $Q$; so let $u$ be the neighbor of $w_1$ and $w_p$ be the neighbor of $v$ so that $u \in \varphi_i$, $v \in \varphi_j$, and $i \leq j$.
    For the base case of $p = k+2$,  we branch our analysis on $s$:
    \begin{itemize}
        \item If $s = 1$, we assign color $1$ to $w_1$ and color $k$ to $w_2$ if $\varphi(u) \neq 1$, otherwise we swap the colors of $w_1$ and $w_2$; we repeat this process with $w_p$ and $w_{p-1}$, depending $\varphi(v)$.
        Now, we have $k-2$ colors to be freely used on $k+2-4 = k - 2$ vertices.
        Since colors $1$ and $k$ were used twice in $P$, the relative difference of their sizes remains the same.
        Since each color $i$ different from $1$ and $k$ was used only once in $P$, the relative difference between $|\varphi_i|$ and the size of $\varphi_k$ increased by at most one; but, since $s = 1$, we had $|\varphi_k| - |\varphi_2| \leq \alpha - 1$, which ensures that the balance $\alpha^*$ of the new coloring is at most $\alpha$.
        \item If $1 < s < k$, for each color $i \in [s]$, pick one vertex of $\{w_1, \dots, w_s\}$ and color it with $i$; since $s > 1$, we can do so and have $\varphi(w_1) \neq \varphi(u)$.
        Do the same for each color $j \in [k] \setminus [s]$ and $\{w_{s+1}, \dots, w_k\}$.
        Finally, color $\{w_{k+1}, w_{k+2}\}$ with colors $1,2$ so $\varphi(w_{k+2}) \neq \varphi(v)$.
        Since $s \neq k$ and each color in $[k] \setminus [2]$ was used once in $P$, the new coloring has the same balance as $\varphi$ unless $k=3$, in which case, $\alpha^* = \alpha-1$.
        \item If $s = k$, color $P$ so each color is used at least once, $\varphi(u) \neq \varphi(w_1)$ and $\varphi(w_{k+2}) = \varphi(v)$, which can be done since $k \geq 3$.
        This implies that the new coloring is $(\alpha+1)$-balanced, since $(k + 2)\! \mod k \not\equiv 0$ whenever $k > 2$, but, since $|\varphi_1| = |\varphi_k|$, $\alpha = 0$.
    \end{itemize}
    
    Before proceeding, we first deal with a corner case when $p = k + 3$, $s=1$, and $u,v \in \varphi_1$: color $w_1$ with $k$, $w_2$ with $1$, $w_3$ with $2$, and $\{w_4, \dots, w_{k+3}\}$ so color $i$ is assigned to $w_{i+3}$; since $s=1$ this maintains the balance of the coloring intact.
    For the general case of $p > k+2$, we again branch on $s$.
    \begin{itemize}
        \item If $s=1$ and $u \notin \varphi_1$, color $w_1$ with $1$ and apply the inductive hypothesis to $Q' = Q \cup \{w_1\}$, $P' = P \setminus \{w_1\}$, and the coloring $\varphi'$ of $Q'$, which is $\alpha'$-balanced.
        Note that $\alpha' = \alpha-1$, so after extending $\varphi'$ to a coloring $\varphi^*$ of $Q \cup P$, we have that the balance of the latter is at most $\alpha^* \leq \alpha'+1 \leq \alpha$.
        If $s=1$ and $u \in \varphi_1$, but $v \notin \varphi_1$, we proceed in a similar manner, but color $w_p$ with $1$.
        \item If $s=1$ and $u,v \in \varphi_1$, we either have $p = k+3$, in which case we have fallen into our corner case, or $p \geq k+4$, in which case we color $w_1$ with $k$ and $w_2$ with $1$, and apply the inductive hypothesis on $Q' = Q \cup \{w_1, w_2\}$, $P' = P \setminus \{w_1, w_2\}$, and on the coloring $\varphi'$ of $Q'$, which is $\alpha$-balanced; note, however, that $\alpha \neq 0$, so by extending $\varphi'$ to a coloring $\varphi
       ^*$ of $Q \cup P$, we have that $\varphi^*$ is $\alpha^*$-balanced with $\alpha^* \leq \alpha$.
       \item If $1 < s < k$, there is at least one color $i \in [s] \setminus \{\varphi(u)\}$, so color $w_1$ with $i$ and apply the inductive hypothesis to $Q' = Q \cup \{w_1\}$, $P' = P \setminus \{w_1\}$, and the coloring $\varphi'$ of $Q'$, which is still $\alpha$-balanced since $s > 1$.
       Once again, the extension $\varphi^*$ of $\varphi'$ to $Q \cup P$ is $\alpha^*$-balanced and, since $\alpha \neq 0$, $\alpha^* \leq \alpha$.
       \item If $s = k$, color $w_1$ with any color different from $\varphi(u)$ and again apply the inductive hypothesis to $Q' = Q \cup \{w_1\}$, $P' = P \setminus \{w_1\}$, and the coloring $\varphi'$ of $Q'$.
       This time, we have that $\varphi'$ is 1-balanced since $s = k$, so the extension $\varphi^*$ of $\varphi$ to $Q \cup P$ is either 0 or 1-balanced.
    \end{itemize}
\end{proof}

We now show how to safely bring an $\alpha$-balanced coloring closer to an equitable coloring, which is either 0 or 1-balanced.
It is worthy to note that, in the next lemma, $c$ is unbounded, i.e it could be the case that $c \geq k+1$.
This is not a problem, however, since we only want to show that, if we have $x$ segments of length $k+1$, we can bring the difference $|\varphi_k| - |\varphi_1|$ down by $x$ units or make the coloring equitable.

\begin{lemma}
    \label{lem:gap_reduce}
    Let $Q$ be an induced subgraph of $G$ that does not contain any vertex of some path $P \in \mathcal{P}_2(G)$ and $\varphi$ be an $\alpha$-balanced coloring of $Q$.
    If $|P| = k+2 + x(k + 1) + c$, where $c \geq 0$ and $x \geq 1$, then it is possible to extend $\varphi$ into an $\alpha^*$-balanced coloring of $Q \cup P$, so that $\alpha^* \leq \max\{\alpha - x, 1\}$.
\end{lemma}

\begin{proof}
    We proceed by induction on $x$, and again define $s$ to be the largest integer that has $|\varphi_1| = |\varphi_s|$, $P = \{w_1, \dots, w_p\}$, $N_Q(w_1) = \{u\}$, and $N_Q(w_p) = \{v\}$.
    To prove the base case when $x=1$, we divide our analysis based on the value of $s$.
    \begin{itemize}
        \item If $s=1$, color $w_1$ with any color $i$ different from $1$ and $\varphi(u)$, assign color $1$ to $w_2$, distribute colors $[k] \setminus \{1, i\}$ to $\{w_3, \dots, w_k\}$ arbitrarily and color $w_{k+1}$ with color $1$; this way we have that this intermediate coloring $\varphi'$ is $(\alpha-1)$-balanced.
        Now, we apply Lemma~\ref{lem:rebalance} to $Q \cup \{w_1, \dots, w_{k+1}\}$, $P \setminus \{w_1, \dots, w_{k+1}\}$, and $\varphi'$ to obtain a coloring $\varphi^*$ of $Q \cup P$ that is $\alpha$-balanced if $\alpha > 1$, or either 0 or 1-balanced if $\alpha = 1$.
        \item If $s > 1$, assign colors in $[s]$ arbitrarily to $\{w_1, \dots, w_s\}$ while guaranteeing $\varphi(u) \neq \varphi(w_1)$, which is possible since $s \geq 2$, and apply Lemma~\ref{lem:rebalance} to $Q \cup \{w_1, \dots, w_{s}\}$, $P \setminus \{w_1, \dots, w_{s}\}$, and $\varphi'$ to obtain a coloring $\varphi^*$ of $Q \cup P$ that is $\alpha$-balanced if $\alpha > 1$, or either 0 or 1-balanced if $\alpha \leq 1$.
    \end{itemize}
    When $x \geq 2$, we repeat the coloring procedures we delineated in the base case to color the vertices of the first section of size $k+1$ of $P$, obtaining the partial coloring $\varphi'$, which is 0-balanced if $\alpha = 0$ (implying $s = k$) or $(\alpha - 1)$-balanced otherwise.
    Now, we apply the inductive hypothesis on the remainder of the path to obtain an $\alpha^*$-balanced coloring $\varphi^*$ that extends $\varphi'$ and has $\alpha
   ^* \leq \max\{(\alpha-1) - (x-1), 1\} = \max\{\alpha - x, 1\}$.
\end{proof}

Now, let $\mathcal{R}(G) = \{P(u_1, v_1), \dots, P(u_r, v_r)\}$ be the set of paths that replaced edges of $H$ and have at least $2k+3$ vertices; in particular, $|P(u_i, v_i)| = k + 2 + c_i + x_i(k+1)$, where $x_i \geq 1$ and $0 \leq c_i \leq k$, i.e $x_i$ is maximum.
Set $\mathcal{X}(G)$ is the collection of all $x_i$'s of paths of $\mathcal{R}(G)$.
Also, let $Q$ be the induced subgraph of $G$ obtained by removing the vertices of the paths in $\mathcal{R}(G)$, and $q = |V(Q)|$.

\begin{lemma}
    \label{lem:compresion}
    If $\sum_{x \in \mathcal{X}(G)} x \geq q$, then $G$ can be equitably $k$-colored if and only if $Q$ can be $k$-colored.
\end{lemma}

\begin{proof}
    If $G$ is equitably $k$-colorable, then $G$ is $k$-colorable and so is every induced subgraph of $G$, including $Q$.
    For the converse, let $\varphi$ be an $\alpha$-balanced $k$-coloring of $Q$.
    By Lemma~\ref{lem:gap_reduce}, with each path $P(u_i, v_i) \in \mathcal{R}(G)$, we can extend $\varphi$ while either reducing the gap between $|\varphi_k|$ and all colors used $|\varphi_1|$ times by $x_i$, or we can obtain a $k$-coloring of $Q \cup P(u_i, v_i)$ that is either 0 or 1-balanced, which is an equitable coloring.
    Since $\sum_{x \in \mathcal{X}(G)} x \geq q \geq \alpha$, we can extend $\varphi$ to $G$ and obtain an equitable $k$-coloring of $G$.
\end{proof}

\begin{corollary}
    \label{cor:kern}
    If $\sum_{x \in \mathcal{X}(G)} x \geq q$, then $G$ can be equitably $k$-colored if and only if the graph obtained by adding $q(k-1)$ isolated vertices to $Q$ can be equitably $k$-colored.
\end{corollary}

\begin{proof}
    For the forward direction, note that any $k$-coloring of an $n$-vertex graph can be made equitable by adding $n(k-1)$ isolated vertices to it.
    For the converse, if $Q \cup I$ admits an equitable $k$-coloring $\varphi'$, its restriction $\varphi$ to $Q$ is $\alpha$-balanced for some $\alpha \leq q$; reasoning as in Lemma~\ref{lem:compresion}, we can 
\end{proof}

\begin{lemma}
    $Q$ has at most $4\ell + (k+1)(11\ell - 2)$ vertices.
\end{lemma}

\begin{proof}
    $|Q|$ is maximized when only one of the $5\ell + \frac{\ell}{2} - 1$ paths that replaced edges of $H$ is in $\mathcal{R}(G)$, while all others have at most $2k+2$ vertices, so $q \leq 4\ell + (5\ell + \frac{\ell}{2} - 1)(2k+2) = 11k\ell + 15\ell - 2k - 2 = 4\ell + (k+1)(11\ell - 2)$.
\end{proof}

\begin{theorem}
    \label{thm:kernel}
    When parameterized by the max leaf number $\ell$ and number of colors $k$, \pname{Equitable Coloring} admits a kernel with $(4\ell + (k+1)(11\ell - 2))(k+2) + 11\ell(k+1)$ vertices.
    If $H$ is given in the input, then the algorithm runs in  polynomial time.
\end{theorem}

\begin{proof}
    If $|V(G)| \leq q(k+2) + 8\ell(k+1)$, we are done.
    On the other hand, if $|V(G)| > q(k+2) + 8\ell(k+1)$, we claim that $\sum_{x \in \mathcal{X}(G)} x \geq q$.
    To see that this is the case, recall that $G \setminus Q = \mathcal{R}(G)$ is a set of $r \leq 5\ell + \frac{\ell}{2}$ disjoint paths, each with at least $2k+3$ vertices, and contains more than $8\ell(k+1) + q(k+1)$ vertices; consequently we have:
    \begin{align*}
        11\ell(k+1) + q(k+1) &\leq \sum_{P_i \in \mathcal{R}(G)} |P_i|\\
        11\ell(k+1) + q(k+1) &\leq r(k+2) + \sum_{i \in [r]}\left(c_i + (k+1)x_i\right)\\
        11\ell(k+1) + q(k+1) &\leq r(k+2) + rk + (k+1)\sum_{i \in [r]} x_i\\
        q &\leq \sum_{i \in [r]} x_i\\
    \end{align*}
    Now, by applying Corollary~\ref{cor:kern}, we have that $Q \cup I$, where $I$ is an independent set with $q(k-1)$ vertices, can be equitably $k$-colored if and only if $G$ is equitably $k$-colorable.
    
    As to the running time, the first step is listing $\mathcal{R}(G)$, which, since we have $H$ in hand, can be done in time linear on $|E(H)| + |V(G)|$: for each edge $uv$ of $H$, list the path $P(u,v)$ in $G$ and check if it has at least $2k+3$ vertices.
    If $|V(G)| \leq q(k+2) + 11\ell(k+1)$ we are done; otherwise we remove the paths of $\mathcal{R}(G)$ and add $I$.
    Since the latter steps can be done in time linear on $|V(G)|$, the kernelization algorithm runs in $\bigO{|E(H)| + |V(G)|}$ time.
\end{proof}

\begin{corollary}
    When parameterized by the max leaf number $\ell$, \pname{Equitable Coloring} admits a kernel with $4\ell(44\ell^2 - 18\ell - 1)$ vertices.
    If $H$ is given in the input, then the algorithm runs in  polynomial time.
\end{corollary}

\begin{proof}
    Since $G$ is a subdivision of a graph on $4\ell$ vertices, the maximum degree $\Delta$ of $G$ is at most $4\ell-1$; by the Hajnal-Szemerédi theorem~\cite{hajnal_szmeredi_theorem}, if $k \geq \Delta + 1$, then $G$ is equitably $k$-colorable.
    The result follows immediately by substituting $k$ with $4\ell - 1$ in the bound of Theorem~\ref{thm:kernel}.
\end{proof}
    \section{A lower bound for vertex cover + number of colors}
\label{sec:vc}

Our final result is a proof that \pname{Equitable Coloring} parameterized by vertex cover and number of colors does not admit a polynomial kernel unless $\NP~\subseteq~\coNP/\poly$.
Our proof makes heavy use of the reduction of Section~\ref{sec:dpath}.
We are going to show: (i) an OR-cross-composition~\cite{cross_composition} from \pname{Multicolored Clique} to \pname{Number List Coloring} parameterized by $\nu + q$, (ii) a PPT reduction from \pname{Number List Coloring} parameterized by $\nu + q$ to \pname{Equitable List Coloring} under the same parameterization, and (iii) a PPT reduction from \pname{Equitable List Coloring} parameterized by vertex cover and number of colors to \pname{Equitable Coloring} parameterized by $\nu + q$.
For step (i), we employ edge gadgets and list assignments based on previous work~\cite{equitable_latin}, although the assignments and the numerical targets suffer adjustments in order to translate the construction from a parameterized reduction to an OR-cross-composition.

Throughout this section, let $\mathcal{H} = \{(H_0, \mathcal{V}_1), \dots, (H_t, \mathcal{V}_t)\}$ be a set of instances of \pname{Multicolored Clique} so that: $V(H_p) = [n]$ and $|\mathcal{V}_p| = k$, for every $p \in [t]$.
Moreover, we may assume that for each $V_i^p, V_j^p \in \mathcal{V}_p$, $|E(V_i^p, V_j^p)| = M$ and $|V_i| = |V_j| = N$, for every $p \in [t]$, where $E(V_i^p, V_j^p)$ denotes the set of edges between $V_i^p$ and $V_j^p$.
We also may safely assume that $|\mathcal{H}|$ is a power of two: if it is not, we may add at most $t$ copies of any instance to $\mathcal{H}$ without changing the outcome of the composition.
We shall denote by $(G, L, h)$ the instance of \pname{Number List Coloring} where $G$ is the input graph, $L$ is a $q$-list-assignment, $h : [q] \mapsto \mathbb{N}$ is a function describing the numerical targets of each color, and the goal is to find an $L$-list-coloring of $G$ so that $|\varphi_i| = h(i)$ for every color $i \in \bigcup_{v \in V(G)}L(v)$.
For the remainder of this section, $q$ is the number of colors.
Intuitively, we are going to interpret each edge $uv \in E(H_p)$ as the pair of \textit{oriented edges} $(u,v)$ and $(v,u)$ and a multicolored clique of size $k$ as a set of $2\binom{k}{2}$ edges.
Each gadget of $G$ will correspond to one such oriented edge and the coloring of the gadget corresponding to $(u,v)$ will tell us if: (i) the edge belongs to the instance $(H_p, \mathcal{V}_p)$ that contains the solution, and (ii) if $(u,v)$ is part of the solution to $(H_p, \mathcal{V}_p)$.

\medskip\noindent \textbf{Construction.} We begin the construction of $(G, L, h)$ by labeling the colors of the set $[q]$ according to role they are going to play in the composition.
As we can see, the number of colors $q$ is a quadratic function of $k$.

\begin{itemize}
    \item[Selection:] The colors $\mathcal{S} = \{\sigma(i,j) \mid (i,j) \in [k]^2, i \neq j\}$ and $\mathcal{S'} = \{\sigma'(i,j) \mid (i,j) \in [k]^2, i \neq j\}$ are used to select which edges must belong to the clique.
    
    \item[Helper:] $\mathcal{Y}$ and $\mathcal{X}$ satisfy $|\mathcal{Y}| = |\mathcal{X}| = |\mathcal{S}|$. These two sets of colors force the choice made at the root of the edge gadgets to be consistent across the gadget.
    
    \item[Symmetry:] The colors $\mathcal{E} = \{\varepsilon(i,j) \mid (i,j) \in [k]^2, i < j\}$  and $\mathcal{E'} = \{\varepsilon'(i,j) \mid (i,j) \in [k]^2, i < j\}$ guarantee that, if edge $e \in E(V_i,V_j)$ is picked from  $V_i$ to $V_j$, it must also be picked from $V_j$ to $V_i$.
    
    \item[Consistency:] Colors $\mathcal{T} = \{\tau_i(r,s) \mid i \in [k],\ r,s \in [k] \setminus \{i\}, \ r < s\}$ and $\mathcal{T'} = \{\tau'_i(r,s) \mid i \in [k],\ r,s \in [k] \setminus \{i\}, \ r < s\}$ ensure that if the edge $uv$ is chosen between $V_i$ and $V_j$, the edge between $V_i$ and $V_r$ must also be incident to $u$.
    
    \item[Suppression:] We have four colors $\mathcal{U} = \{\alpha, \gamma, \rho, \lambda\}$ whose purpose is to color gadgets that correspond to edges not present in the graph $H_p$ that contains the solution.
    
    \item[Shading:] For each Selection, Helper, Symmetry, Consistency, and Suppression color $c$, we have an additional color $\conj{c}$, which is used in the instance selector gadget to remove the influence of some vertices in the coloring of the edge gadgets.
    
    \item[Propagation:] A single color $\beta$ is used to enforce that choices in a part of the instance selector gadget are propagated to the whole gadget.
\end{itemize}

As in~\cite{equitable_latin,colorful_treewidth}, let $Z$ be a huge integer, say $Z = n^3$.
The \textit{up-identification number} for vertex $v \in [n]$, denoted by $v_\uparrow$, is a unique number in the interval $[n^2 + 1, n^2 + n]$, while the \textit{down-identification number} $v_\downarrow$ for vertex $v$ is given by $v_\downarrow = Z - v_\uparrow$.
The up-identification number $e_\uparrow$ for edge $e$ is defined similarly, but taken from the interval $\left[2n^2 + 1, 2n^2 + |\bigcup_{p \in [t]} E(H_p)|\right]$, while the down-identification number is defined by $e_\downarrow = Z - e_\uparrow$.

Moving on to the gadgets themselves, we say that a tuple $(i,j,u,v)$ belongs to $(H_p, \mathcal{V}_p)$ if $i,j \in [k]$, $uv \in E(V_i^p, V_j^p)$, and $u \in V_i^p$; note that there are at most $(kn)^2$ possible tuples, regardless of how many instances we have in $\mathcal{H}$.

\medskip\noindent \textbf{Forward Edge Gadget.}
For each $p \in [t] \cup \{0\}$ and each tuple $(i,j,u,v) \in (H_p, V_p)$ where $i < j$, we add to $G$ a forward edge gadget, which we denote by $\overrightarrow{G}(i,j,u,v)$.
This gadget has a root vertex $r(i, j, u, v)$, with list $\{\sigma(i,j), \sigma'(i, j), \alpha\}$, and two neighbors, both with the list $\{\sigma(i,j), y(i,j), \rho\}$, which for convenience we label as $a(i,j,u,v)$ and $b(i,j,u,v)$.
We equate membership of edge $e$ in the solution to one of the instances of \pname{Multicolored Clique} to the coloring of $r(i,j,u,v)$ with $\sigma(i,j)$.
When discussing the remaining vertices of the gadget, we say that a vertex is \textit{even} if its distance to $r(i,j,u,v)$ is even, otherwise it is \textit{odd}.
To $a(i,j,u,v)$, we append a path with $2e_\down + 2(k-1)u_\down$ vertices.
First, we choose $e_\down$ even vertices to assign the list $\{y(i, j), \varepsilon'(i, j), \lambda\}$.
Next, for each $r$ in $j < r \leq k$, choose $u_\down$ even vertices to assign the list $\{y(i,j), \tau'_i(j,r), \lambda\}$.
Similarly, for each $s \neq i$ satisfying $s < j$, choose $u_\down$ even vertices and assign the list $\{y(i,j), \tau_i(s,j), \lambda\}$.
All the odd vertices - except $a(i,j,u,v)$ and $b(i,j,u,v)$ - are assigned the list $\{y(i,j), x(i,j), \gamma\}$.
The path appended to $b(i,j,u,v)$ is similarly defined, except for two points: (i) the length and number of chosen vertices are proportional to $e_\up$ and $u_\up$; and (ii) when color $\varepsilon(i,j)$ (resp. $\tau_i(s,r)$) should be in the list, we the list contains $\varepsilon'(i,j)$ (resp. $\tau'_i(s,r)$) instead, and vice-versa.
For an example of the forward edge gadget, please refer to Figure~\ref{fig:for_edge_gadget2}.

\begin{figure}[!htb]
    \centering
     \begin{tikzpicture}[scale=0.8]
         \GraphInit[unit=3,vstyle=Normal]
         \SetVertexNormal[Shape=circle, FillColor=black, MinSize=3pt]
         \tikzset{VertexStyle/.append style = {inner sep = \inners, outer sep = \outers}}
         \SetVertexLabelOut
         \Vertex[a=90, d=3.5, Lpos=90, Math, L={\{\sigma(1,3), \sigma'(1,3), \alpha\}}]{r}
         \Vertex[a=72, d=2.5, NoLabel]{b}
         \node at (0, 2) {$\{\sigma(1,3), y(1,3), \rho\}$};
         \foreach \i in {0,1,2,3,4} {
            \pgfmathsetmacro{\ae}{72 - 15.3*(2*\i+1)}
            \pgfmathsetmacro{\ao}{\ae -15.3}
            \Vertex[a=\ao, d=3.5, NoLabel]{ou\i}
         }
         \foreach \i in {0} {
            \pgfmathsetmacro{\ae}{72 - 15.3*(2*\i+1)}
            \newcommand{\lab}{\varepsilon(1,3)}
            \Vertex[a=\ae, d=3.5, Lpos=\ae, Math, L={\{y(1,3), \lab, \lambda\}}]{eu\i}
            \Edge(eu\i)(ou\i)
         }
         \foreach \i in {1,2} {
            \pgfmathsetmacro{\ae}{72 - 15.3*(2*\i+1)}
            \newcommand{\lab}{\tau_1(3,4)}
            \Vertex[a=\ae, d=3.5, Lpos=\ae, Math, L={\{y(1,3), \lab, \lambda\}}]{eu\i}
            \Edge(eu\i)(ou\i)
         }
         \foreach \i in {3,4} {
            \pgfmathsetmacro{\ae}{72 - 15.3*(2*\i+1)}
            \newcommand{\lab}{\tau'_1(2,3)}
            \Vertex[a=\ae, d=3.5, Lpos=\ae, Math, L={\{y(1,3), \lab, \lambda\}}]{eu\i}
            \Edge(eu\i)(ou\i)
         }
         
         \foreach \i in {0,1,2,3} {
            \pgfmathtruncatemacro{\n}{\i+1}
            \Edge(eu\n)(ou\i)
         }
         
         \node at (-5.5,0) {$\{y(1,3), x(1,3), \gamma\}$};
         
         \Vertex[a=108, d=2.5, NoLabel]{a}
         \foreach \i in {0,1,2,3} {
            \pgfmathsetmacro{\ae}{108 + 18*(2*\i+1)}
            \pgfmathsetmacro{\ao}{\ae + 18}
            \Vertex[a=\ao, d=3.5, NoLabel]{od\i}
         }
         \foreach \i in {0,1} {
            \pgfmathsetmacro{\ae}{108 + 18*(2*\i+1)}
            \newcommand{\lab}{\varepsilon'(1,3)}
            \Vertex[a=\ae, d=3.5, Lpos=\ae, Math, L={\{y(1,3), \lab, \lambda\}}]{ed\i}
            \Edge(ed\i)(od\i)
         }
         \foreach \i in {2} {
            \pgfmathsetmacro{\ae}{108 + 18*(2*\i+1)}
            \newcommand{\lab}{\tau'_1(3,4)}
            \Vertex[a=\ae, d=3.5, Lpos=\ae, Math, L={\{y(1,3), \lab, \lambda\}}]{ed\i}
            \Edge(ed\i)(od\i)
         }
         \foreach \i in {3} {
            \pgfmathsetmacro{\ae}{108 + 18*(2*\i+1)}
            \newcommand{\lab}{\tau_1(2,3)}
            \Vertex[a=\ae, d=3.5, Lpos=\ae, Math, L={\{y(1,3), \lab, \lambda\}}]{ed\i}
            \Edge(ed\i)(od\i)
         }
         
         \foreach \i in {0,1,2} {
            \pgfmathtruncatemacro{\n}{\i+1}
            \Edge(ed\n)(od\i)
         }
         \Edges(ed0, a, r, b, eu0)
         \draw[->,thick,>=stealth] (-1.5, 1.6) -- (-1.5,-2.5);
         \draw[->,thick,>=stealth] (1.5,-2.5) -- (1.5, 1.6);
     \end{tikzpicture}
    \caption{Example of a forward edge gadget $\protect\overrightarrow{G}(1,3,u,v)$ with $k = 4$, $Z = 3$, $e_\down = 2$, and $u_\down = 1$. Vertices with no explicit list have list equal to $\{y(1,3), x(1,3), \gamma\}$.}
    \label{fig:for_edge_gadget2}
\end{figure}
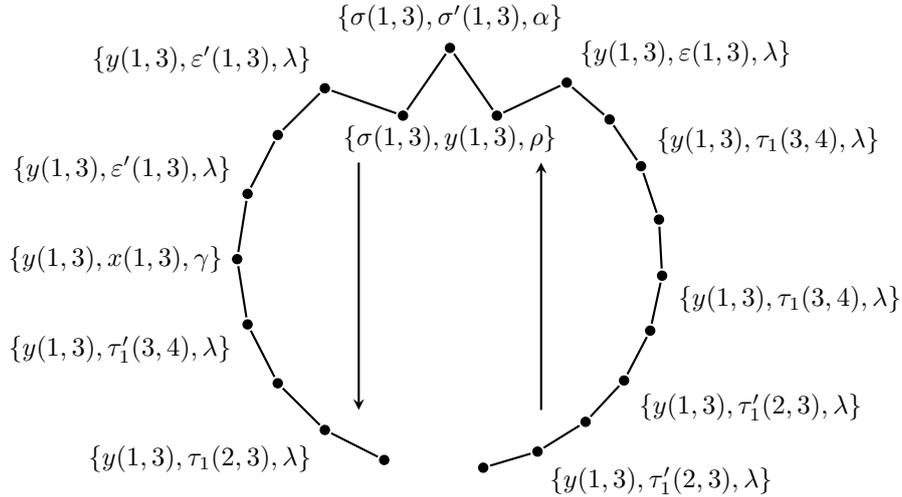

\medskip\noindent \textbf{Backward Edge Gadget.}
For each $p \in [t] \cup \{0\}$, once again, let $(i,j,u,v) \in (H_p, \mathcal{V}_p)$ be a tuple where $i < j$.
We add to $G$ a backward edge gadget, denoted by $\overleftarrow{G}(j,i,v,u)$.
This gadget has vertices $r(j, i, v,u)$, $a(j,i,v,u)$, and $b(j,i,v,u)$ defined similarly as to the forward gadget, with the root vertex having the list $\{\sigma(j,i), \sigma'(j,i), \alpha\}$, while the other two have the list $\{\sigma(j,i), y(j,i), \rho\}$.
To $a(j,i,v,u)$, we append a path with $2e_\down + 2(k-1)v_\down$ vertices.
First, choose $e_\down$ even vertices to assign the list $\{y(j,i), \varepsilon(i, j), \lambda\}$. 
Now, for each $r$ in $j < r \leq k$, choose $v_\down$ even vertices to assign the list $\{y(j,i), \tau'_j(i,r), \lambda\}$.
Then, for each $s \neq i$ satisfying $s < j$, choose $v_\down$ even vertices and assign the list $\{y(j,i), \tau_j(s,i), \lambda\}$.
All the odd vertices are assigned the list $\{y(j,i), x(j,i), \gamma\}$.
The path appended to $b(j,i,v,u)$ is similarly defined, except that: (i) the length and number of chosen vertices are proportional to $e_\up$ and $v_\up$; and (ii) when the color $\varepsilon(i,j)$ (resp. $\tau_j(s,r)$) is in the list, we replace it with $\varepsilon'(i,j)$ (resp. $\tau'_j(s,r)$), and vice-versa.
Note that, for every edge gadget, either forward or backward, the number of vertices is equal to $3 + 2(e_\up + e_\down) + 2(k-1)(u_\up + u_\down) = 3 + 2kZ$. 
We say that an edge gadget is \textit{selected} if its root vertex is colored with $\sigma(i,j)$, we say that it is \textit{passed} if the root is colored with $\sigma'(i,j)$, otherwise we say that it is \textit{suppressed}.

\medskip\noindent \textbf{Instance Selector Gadget.}
We now move on to a more delicate piece of the composition: the instance selector gadget.
For each color $c$ that is neither a Shading color nor a Propagation color, we add a choice gadget $\mathcal{C}(c)$ which contains a perfect  matching $B_c = \{b(c,1), \dots, b(c,s)\}$, where $s = \log_2 |\mathcal{H}| = \log_2 (t+1)$ and $b(c, i) = \{(c,i)_0, (c,i)_1\}$ is an edge of the matching, and an independent set $I_c = \{z_0(c), \dots, z_t(c)\}$.
We say that $b(c, i)$ is the $i$-th bit of gadget $\mathcal{C}(c)$ and that the $i$-th bit is equal to $\mu \in \{0,1\}$ if $(c,i)_\mu$ is colored with $\beta$. 
We connect $z_p(c)$ to $B_c$ as follows: if the $i$-th least significant bit of $p$ is $0$, connect $z_p(c)$ to $(c,i)_0$, otherwise connect it to $(c,i)_1$.
As to the lists, we define $L((c, j)_0) = L((c, j)_1) = \{c, \beta\}$ and $L(z_p(c)) = \{c, \conj{c}\}$.
Intuitively, the vertices of the matching that are colored with $\beta$ encode the binary representation of a unique $p \in [t] \cup \{0\}$, which implies that only $z_p(c)$ may be colored with $c$, all other vertices of $I_c$ must be colored with $\conj{c}$.
This index $p$ will be precisely the index of the instance $(H_p, \mathcal{V}_p)$ that will be a \YES\ instance of \pname{Multicolored Clique}.

Naturally, our goal is to force that the same $p$ is encoded by all choice gadgets, i.e. that the same instance $(H_p, \mathcal{V}_p)$ is picked for all colors.
To do so, we pick any choice gadget $\mathcal{C}(c)$ to be a primary gadget, and, for every other choice gadget $\mathcal{C}(d)$, we add the edges $(c,i)_0(d,i)_1$ and $(c,i)_1(d,i)_0$; this way, the $i$-th bit of $\mathcal{C}(d)$ has $(d,i)_0$ colored with $\beta$ if and only if $(c,i)_0$ is colored with $\beta$.
The union of all these choice gadgets is our instance selector gadget, which we present an example of in Figure~\ref{fig:sel_gadget}.

\begin{figure}[!htb]
    \centering
        \begin{tikzpicture}[scale=1]
            \GraphInit[unit=3,vstyle=Normal]
            \SetVertexNormal[Shape=circle, FillColor=black, MinSize=3pt]
            \tikzset{VertexStyle/.append style = {inner sep = \inners, outer sep = \outers}}
            \begin{scope}
                \node at (-3, 2.4) {$\{c, \conj{c}\}$};
                \node at (-1, 2.4) {$\{c, \beta\}$};
                \node at (3, 2.4) {$\{d, \beta\}$};
                \node at (5, 2.4) {$\{d, \conj{d}\}$};
                \SetVertexLabelOut
                \Vertex[x=-3, y=1.725, Lpos=180,Math, L={z_3(c)}]{x_0}
                \Vertex[x=-3, y=0.575, Lpos=180,Math, L={z_2(c)}]{x_1}
                \Vertex[x=-3, y=-0.575, Lpos=180,Math, L={z_1(c)}]{x_2}
                \Vertex[x=-3, y=-1.725, Lpos=180,Math, L={z_0(c)}]{x_3}
                
                \SetVertexNoLabel
                \begin{scope}[xshift=-1cm, rotate=90]
                    \Vertex[x=-1.725, y=0, Lpos=0,Math]{al_0}
                    \Vertex[x=-1.275, y=0, Lpos=0,Math]{be_0}
                    \draw (-1.10, -0.25) rectangle (-1.9, 0.25);
                    
                    \Vertex[x=1.275, y=0, Lpos=0,Math]{al_1}
                    \Vertex[x=1.725, y=0, Lpos=0,Math]{be_1}
                    \draw (1.10, -0.25) rectangle (1.9, 0.25);
                \end{scope}
                \Edges[style={opacity=0.2}](be_0, x_0, be_1)
                \Edges[style={opacity=0.2}](be_0, x_1, al_1)
                \Edges[style={opacity=0.2}](al_0, x_2, be_1)
                \Edges[style={opacity=0.2}](al_0, x_3, al_1)
                \begin{scope}
                    \tikzset{VertexStyle/.append style = {shape = rectangle, inner sep = 2.2pt}}
                    \AddVertexColor{black}{be_0, be_1}
                    \AddVertexColor{white}{al_0, al_1}
                \end{scope}
            \end{scope}
            
            \begin{scope}[xshift=2cm, rotate=180]
                \SetVertexLabelOut
                \Vertex[x=-3, y=1.725, Lpos=0,Math, L={z_0(d)}]{dx_0}
                \Vertex[x=-3, y=0.575, Lpos=0,Math, L={z_1(d)}]{dx_1}
                \Vertex[x=-3, y=-0.575, Lpos=0,Math, L={z_2(d)}]{dx_2}
                \Vertex[x=-3, y=-1.725, Lpos=0,Math, L={z_3(d)}]{dx_3}
                
                \SetVertexNoLabel
                \begin{scope}[xshift=-1cm, rotate=90]
                    \Vertex[x=-1.725, y=0, Lpos=0,Math]{dal_1}
                    \Vertex[x=-1.275, y=0, Lpos=0,Math]{dbe_1}
                    \draw (-1.10, -0.25) rectangle (-1.9, 0.25);
                    
                    \Vertex[x=1.275, y=0, Lpos=0,Math]{dal_0}
                    \Vertex[x=1.725, y=0, Lpos=0,Math]{dbe_0}
                    \draw (1.10, -0.25) rectangle (1.9, 0.25);
                \end{scope}
                \Edges[style={opacity=0.2}](dbe_0, dx_0, dbe_1)
                \Edges[style={opacity=0.2}](dbe_0, dx_1, dal_1)
                \Edges[style={opacity=0.2}](dal_0, dx_2, dbe_1)
                \Edges[style={opacity=0.2}](dal_0, dx_3, dal_1)
                \begin{scope}
                    \tikzset{VertexStyle/.append style = {shape = rectangle, inner sep = 2.2pt}}
                    \AddVertexColor{white}{dbe_0, dbe_1}
                    \AddVertexColor{black}{dal_0, dal_1}
                \end{scope}
            \end{scope}
            \Edge(be_0)(al_0)
            \Edge(be_1)(al_1)
            \Edge(dbe_0)(dal_0)
            \Edge(dbe_1)(dal_1)
            
            \Edge(be_0)(dbe_0)
            \Edge(be_1)(dbe_1)
            \Edge(al_0)(dal_0)
            \Edge(al_1)(dal_1)
        \end{tikzpicture}
    \caption{Instance selector gadget for $\log_2 |\mathcal{H}| = 2$. Black rectangles are vertices corresponding to bit value 1; white rectangles correspond to bit value 0.\label{fig:sel_gadget}}
\end{figure}
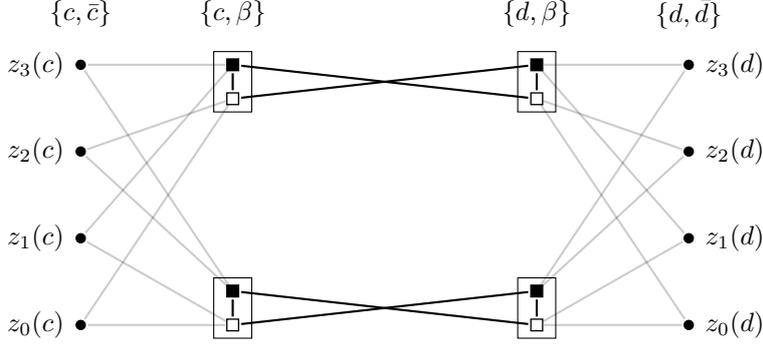

To conclude the construction of $G$, let $c$ be a color that has an associated choice gadget $\mathcal{C}(c)$, $z_p(c)$ be a vertex of the independent set $I_c$, and $G(i,j,u,v)$ be an edge gadget (either forward or backward).
If $c$ is a Suppression color and $uv \in E(V_i^p, V_j^p)$ (suppose $u \in V_i^p$), then add an edge between $z_p(c)$ and every vertex in $G(i,j,u,v)$ that has $c$ in its list.
Otherwise, if $c$ is \textbf{not} a Suppression color and $e \notin E(V_i^p, V_j^p)$, then add an edge between $z_i(c)$ and every vertex in $G(i,j,u,v)$.

\medskip\noindent \textbf{Numerical Targets.}
The final ingredients are the numerical targets.
Recall that $q$ is the number of colors, that $M = |E(V_i^p, V_j^p)|$ for every triple $i,j,p$, that $k = |\mathcal{V}_p|$ for every $p \in [t] \cup \{0\}$, and that $Z$ is the suitable huge integer of the order of $n^3$.
For simplicity, we define $W = \log_2 |\mathcal{H}|$.
The numerical targets are defined according to the following list;
the term $W + 1$ in most targets comes from the fact that, for each color $c$ that is neither a Shading color nor the Propagation color, we must use $c$ exactly $W + 1$ times in the instance selector gadget.
Intuitively, the numerical targets for these colors fulfill the exact the same roles as in Section~\ref{sec:dpath}, the extra $W+1$ term is present so we can appropriately color the instance selector gadget. 

\begin{itemize}
    \item[Selection:] $h(\sigma(i,j)) = 1 + 2(M-1) + W + 1 = 2M + W$ and $h(\sigma'(i,j)) = M - 1 + W + 1 = M+W$.
    Since only one edge may be chosen from $V_i^p$ to $V_j^p$, the non-selection color $\sigma'(i,j)$ must be used in $M-1$ edges of $\mathcal{G}(i,j)$.
    Thus, exactly one $G(i,j,u,v)$ is selected and, to achieve the target of $1 + 2(M-1)$, for every $fg \in E(V_i^p, V_j^p) \setminus \{uv\}$, both $a(i,j,f,g)$ and $b(i,j,f, g)$ must also be colored with $\sigma(i,j)$.
    
    \item[Helper:] $h(y(i,j)) = 2 + kMZ + W + 1 = 3 + kMZ + W$ and $h(x(i,j)) = kMZ - kZ + W + 1 = (M-1)kZ + W + 1$.
    The goal here is that, if a forward or backward gadget $G(i,j,u,v)$ is selected, all the odd positions must be colored with $y(i,j)$, otherwise every even position must be colored with it.
    
    \item[Symmetry:] $h(\varepsilon(i,j)) = h(\varepsilon'(i,j)) = Z + W + 1$.
    If the previous condition holds and $r(i,j,u,v)$ is colored with $\sigma(i,j)$, then $\varepsilon(i,j)$ appears in $e_\up$ vertices of the gadget rooted at $r(i,j,u,v)$.
    To meet the target $Z$, $e_\down$ vertices of another gadget must also be colored with it; as we show, the only way is if $r(j,i,v,u)$ is colored with $\sigma(j,i)$.
    
    \item[Consistency:] Simply set $h(\tau_i(s,r)) = h(\tau'_i(s,r)) = Z + W + 1$.
    Similar to symmetry colors.
    
    \item[Suppression:] Let $Q = \left|\bigcup_{p \in [t] \cup \{0\}} \left\{(i,j,u,v) \mid (i,j,u,v) \in (H_p, \mathcal{V}_p)\right\}\right|$, i.e. the number of tuples $(i,j,u,v)$ that belong to some of the $t+1$ instances of \pname{Multicolored Clique}, which equals the number of edge gadgets we added to $G$.
    We first set $h(\alpha) = Q - 2\binom{k}{2}M + W + 1$: we want that every gadget corresponding to a tuple $(i,j,u,v)$ not in the chosen instance be suppressed, i.e. have the root colored with $\alpha$.
    Then, we set $h(\rho) = 2(Q - 2\binom{k}{2}M) + W + 1$: the same gadgets that have the root colored with $\alpha$ must have the neighbors of the root colored with $\rho$.
    Afterwards, we set $h(\lambda) = h(\gamma) = (Q - 2\binom{k}{2}M)kZ + W + 1$: each gadget whose root was colored by $\alpha$ will have every vertex, except for the root and its neighbors, colored with $\lambda$ or $\gamma$; since the gadget is a path and has $2kZ + 3$ vertices, each suppressed gadget will contribute with $kZ$ vertices colored with $\lambda$ and $kZ$ colored with $\gamma$.
    
    \item[Shading:] For a Selection, Helper, Symmetry, Consistency, or Suppression color $c$, set $h(\conj{c}) = t - 1$: color $\conj{c}$ can only by used in the independent set $I_c$, and $\conj{c}$ must be used in all but one vertex of $I_c$, since only one $z_p(c)$ will not be adjacent to a vertex of $B_c$ colored with $c$.
    
    \item[Propagation:] Finally, we set $h(\beta)= \frac{q-1}{2}W$, since each of the $\frac{q-1}{2}$ choice gadgets has exactly $W$ vertices that must be colored with $\beta$.
\end{itemize}

We now show that the composition is correct.
We being with the auxiliary Lemma~\ref{lem:default_form}, as it presents many of the ideas required for the proof of Theorem~\ref{thm:no_kernel_nlc}.

\begin{lemma}
    \label{lem:default_form}
    Let $\varphi$ be a list coloring of $G$ that meets the numerical targets for Shading colors, $c$ be the color corresponding to the primary choice gadget, and $z_p(c) \in I_c$ be such that $\varphi(z_p(c)) = c$.
    Then, the following conditions hold:
    \begin{enumerate}
        \item{\label{cond:unity}} For every choice gadget $\mathcal{C}(d)$, $\varphi(z_p(d)) = d$.
        
        \item The Propagation color $\beta$ is used exactly $h(\beta)$ times.
        
        \item If $d \neq \beta$ and $d$ is not a Shading color, then $d$ is used in $W+1$ vertices of the instance selector gadget.
        
        \item The Suppression colors can be used on a (forward or backward) edge gadget $G(i,j,u,v)$ if and only if $(i,j,u,v) \notin (H_p, \mathcal{V}_p)$.
        
        \item Suppression colors are the only colors that can be used on a (forward or backward) edge gadget $G(i,j,u,v)$ where $(i,j,u,v) \notin (H_p, \mathcal{V}_p)$.
        
        \item For every Suppression color $d$, $|\varphi_d| = h(d)$.
    \end{enumerate}
\end{lemma}

\begin{proof}
    Note that we do not assume that $\varphi$ satisfies every numerical target given by $h$, it is simply a proper coloring of $G$ that respects the lists given by $L$.
    We prove the conditions in the order they were given in the statement.
    \begin{enumerate}
        \item By construction, the $i$-th bit of $\mathcal{C}(c)$ is equal to $0$ if and only if the $i$-th bit of the binary representation of $p$ is equal to $0$.
        Moreover, for every other color $d$ with a choice gadget, the $i$-th bit of $\mathcal{C}(d)$ cannot be equal to $1$, otherwise, $(c,i)_0(d,i)_1$ would be a monochromatic edge.
        Thus, $z_p(d)$ is the unique vertex of $I_d$ that does not have a neighbor colored with $d$ and, since, $|\varphi_\conj{d}| = h(\conj{d})$, $\varphi(z_p(d)) = d$.
        
        \item This follows directly from the fact that there are $\frac{q-1}{2}$ choice gadgets and, in each $\mathcal{C}(d)$, $\beta$ can only be used in the perfect matching $B_d$ and must be used $\frac{|B_d|}{2} = \log_2 |\mathcal{H}| = W$.
        
        \item Note that that $d$ is used in every vertex of $B_d$ not colored with $\beta$ and there is $z_p(d)$ is the unique vertex of $I_d$ where $\varphi(z_p(d)) = d$.
        Since no other choice gadget contains a vertex with $d$ on their list, $d$ is used $\frac{|B_d|}{2} + 1 = W + 1$ times.
        
        \item If a Suppression color $d$ can be used on $G(i,j,u,v)$, then $z_p(d)$ is not adjacent to any vertex of $G(i,j,u,v)$, which, by construction, implies the $(i,j,u,v) \notin (H_p, \mathcal{V}_p)$.
        For the converse, if $(i,j,u,v) \notin (H_p, \mathcal{V}_p)$ then $z_p(d)$ is not adjacent to the vertices of $G(i,j,u,v)$, which implies that $d$ can be used in $G(i,j,u,v)$.
        By condition~\ref{cond:unity}, these arguments hold for every Suppression color.
        
        \item Let $g$ be a vertex of $G(i,j,u,v)$, $L(g) = \{d,e,f\}$, and $d$ be the unique suppression color in $L(v)$.
        By the construction of $G$, $v$ is adjacent to $z_p(e)$ and $z_p(f)$ which, again by condition~\ref{cond:unity}, satisfy $\varphi(z_p(e)) = e$ and $\varphi(z_p(f)) = f$.
        Moreover, $gz_p(d) \notin E(G)$, every neighbor of $g$ in $\mathcal{C}(d)$ is colored with $\conj{d}$, and no neighbor of $g$ in $G \setminus \mathcal{C}(d)$ has color $d$ on their list, so $g$ can only be colored with $d$.
        
        \item Once again, let $Q = \left|\bigcup_{p \in [t] \cup \{0\}} \left\{(i,j,u,v) \mid (i,j,u,v) \in (H_p, \mathcal{V}_p)\right\}\right|$.
        By the previous conditions, Suppression color $\alpha$ is used $W + 1 = \log_2 |\mathcal{H}| + 1$ times in the instance selector gadget and, for every $(i,j,u,v) \notin (H_p, \mathcal{V}_p)$, $\alpha$ is used once; equivalently, for every element of $S = \{(i,j,u,v) \mid (i,j,u,v) \in (H_p, \mathcal{V}_p\}$, $\alpha$ is not used.
        Thus, $\alpha$ is used in $Q - |S| + W + 1 = Q - 2\binom{k}{2}M + W + 1 = h(\alpha)$ vertices.
        For color $\rho$, note that $\rho$ can must be used on $a(i,j,u,v)$ and $b(i,j,u,v)$ if $r(i,j,u,v)$ is colored with $\alpha$, so $\rho$ is used twice as many times on the edge gadgets, i.e. $|\varphi_\rho| = 2|\varphi_\alpha| - (W + 1) = 2(Q - 2\binom{k}{2}M) + W + 1 = h(\rho)$.
        Finally, color $\lambda$ is used on half of the $2kZ$ vertices of $G(i,j,u,v) \setminus N[r(i,j,u,v)]$, and $\gamma$ on the other half, whenever $r(i,j,u,v)$ is colored with $\alpha$.
        Consequently, $|\varphi_\lambda| = |\varphi_\gamma| = |\varphi_\alpha|kZ - (kZ - 1)(W+1) = (Q~-~2\binom{k}{2}M)kZ~+~W~+~1 = h(\lambda) = h(\gamma)$.
    \end{enumerate}
\end{proof}

\begin{lemma}
    \label{lem:forward_nlc_kernel}
    If $(H_p, \mathcal{V}_p)$ admits a solution $Q$, then $(G, L, h)$ admits a list coloring that satisfies by $h$.
\end{lemma}

\begin{proof}
    Let $Q$ be the solution to $(H_p, \mathcal{V}_p)$ and be $c$ the color corresponding to the primary choice gadget.
    We begin with the following partial coloring $\varphi'$ of $G$:
    \begin{itemize}
        \item For every choice gadget $\mathcal{C}(d)$, set $\varphi'(z_p(d)) = d$ and, for every $z_q(d) \in I_c \setminus \{z_p(d)\}$, set $\varphi'(z_q(d)) = \conj{d}$.
        \item Color the $i$-th bit of $B_d$ so that $\varphi'((d,i)_0) = \beta$ if and only if the $i$-th least significant bit of $p$ if $0$; the other vertex of $b(d,i)$ is colored with $d$.
        \item For every tuple $(i,j,u,v) \notin (H_p, \mathcal{V}_p)$ that has an associated edge gadget $G(i,j,u,v)$, color its vertices using the unique Suppression color present on their list.
    \end{itemize}
    
    \begin{claim}
        The partial coloring $\varphi'$ is a proper partial coloring of $G$.
    \end{claim}
    
    \begin{cproof}
        First, no two adjacent vertices of the instance selector gadget have the same color since: (i) each bit $b(d,i) \in B_d$ has one vertex of each color, (ii) every neighbor of $z_p(d)$ is colored with $\beta$, (iii) no neighbor of $z_s(d) \in I_d \setminus \{z_s(d)\}$ has $\conj{d}$ on their list, and (iv) the vertex of $b(c,i)$ colored with $\beta$ is adjacent to the vertex of $b(d,i)$ that is colored with $d$.
        As to the colored edge gadgets, each pair of adjacent vertices in such a gadget $G(i,j,u,v)$ has a different Suppression color on their list; since these were the colors currently used, no monochromatic edge may exist in $G(i,j,u,v)$.
        Finally, no $z_p(d)$ is adjacent to a vertex of $G(i,j,u,v)$ colored with $d$, since we only used Suppression colors in $G(i,j,u,v)$ and, when $d$ is a Suppression color, $z_p(d)$ is adjacent to edge gadgets where $(i,j,u,v) \in (H_p, \mathcal{V}_p)$.
    \end{cproof}
    
    Using a similar analysis to Lemma~\ref{lem:default_form}, we conclude that $|\varphi'_d| = h(d)$ for every Suppression, Shading, and Propagation color $d$ and, for all other colors $f$, we must use $f$ exactly $h(f) - (W + 1)$.
    Moreover, the only edge gadgets $G(i,j,u,v)$ that remain uncolored are those where $(i,j,u,v) \in (H_p, \mathcal{V}_p)$, and no Suppression color $d$  can be used in vertices of $G(i,j,u,v)$ since $z_p(d)$ is adjacent to vertices of $G(i,j,u,v)$ if and only if $(i,j,u,v) \in (H_p, \mathcal{V}_p)$.
    We proceed to extend $\varphi'$ into a coloring $\varphi$ of $G$ as follows: for each edge $uv \in Q$, we color the root of $G(i,j,u,v)$ with $\sigma(i,j)$, each odd vertex with $y(i,j)$, and every even vertex with the only remaining color; for all other still uncolored gadgets, we color the root with $\sigma'(i,j)$, every even vertex with $y(i,j)$ and all odd vertices with the only possible color left on their lists.
    The following claim concludes the proof of this lemma.

    \begin{claim}
        The coloring $\varphi$ is a proper coloring of $G$ that meets the numerical targets given by $h$.
    \end{claim}
    
    \begin{cproof}
        First, note that each Selection color $d \in \mathcal{S} \cup \mathcal{S}'$ is used the appropriate number of times since: (i) it is used $W+1$ times in the instance selector gadget, (ii) no gadget $G(i,j,u,v)$ where $(i,j,u,v) \notin (H_p, \mathcal{V}_p)$ can be colored with $d$, and (iii) only one edge $uv \in E(V_i^p, V_j^p)$ belongs to $Q$.
        For each color $y(i,j) \in \mathcal{Y}$ and edge $uv$ with $u \in V_i^p$, gadget $G(i,j,u,v)$ has, at least, $e_\up + e_\down + (k-1)(u_\up + u_\down) = kZ$ vertices colored with $y(i,j)$, since either all odd vertices or even vertices are colored with it.
        For the remaining $W+3$ uses of $y(i,j)$, note that the selected gadget $G(i,j,u,v)$ has $a(i,j,u,v)$ and $b(i,j,u,v)$ also colored with $y(i,j)$, and $\mathcal{C}_{y(i,j)}$ has $W+1$ vertices colored with $y(i,j)$.
        As to the other helper color, $x(i,j)$, we use it only on $\mathcal{C}_{x(i,j)}$ and on passed gadgets -- in this case, in every odd vertex (except the $a$'s and $b$'s); this sums up to $W + 1 \sum_{e \in E(V_i^p, V_j^p) \setminus Q} e_\up + e_\down + (k-1)(u_\up + u_\down) = kMZ - kZ + W + 1 = h(x(i,j))$.
        In terms of symmetry colors, $\varphi$ only uses $\varepsilon(i,j)$ on $\mathcal{C}_{\varepsilon(i,j)}$ and on the selected gadgets $G(i,j,u,v)$ and $G(j,i,v,u)$ ($i < j$); in particular, $\varepsilon(i,j)$ is used $e_\down$ times in $G(j,i,v,u)$ and $e_\up$ times on $G(i,j,u,v)$, so it holds that $|\varphi_{\varepsilon(i,j)}| = e_\up + e_\down + W + 1 = Z + W + 1$.
        Note that the same argument applies for color $\varepsilon'(i,j)$.
        Finally, for consistency colors, note that $\tau_i(r,s)$ is used only on $\mathcal{C}_{\tau_i(r,s)}$ and on the selected gadgets $G(i,r,u,v)$ and $G(i,s,u,w)$, specifically, if $i < r < s$, $\tau_i(r,s)$ is used $u_\up$ times in $G(i,r,u,v)$ and $u_\down$ times in $G(i,s,u,w)$, since edges $uv,uw$ must be incident to the same vertex of $Q \cap V_i^p$.
        Consequently $\tau_i(r,s)$ is used in $u_\up + u_\down + W + 1 = Z + W + 1$ vertices.
        The reasoning for $\tau'_i(r,s)$ is similar.
        For all other colors, we had already established that their numerical targets had already been met; since no additional vertex used any of them, it follows that $\varphi$ meets all numerical targets and is a proper coloring. 
    \end{cproof}
\end{proof}

\begin{lemma}
    \label{lem:backward_nlc_kernel}
    If $G$ admits a list coloring $\varphi$ that respects $L$ and satisfies the numerical targets $h$, then at least one instance $(H_p, \mathcal{V}_p)$ of $\mathcal{H}$ admits a solution.
\end{lemma}

\begin{proof}
    Let $\mathcal{C}(c)$ be the primary choice gadget of $G$ and $z_p(c)$ be the unique vertex of $I_c$ where $\varphi(z_p(c)) = c$.
    We are going to show that $(H_p, \mathcal{V}_c)$ contains a multicolored clique of appropriate size.
    As a consequence of Lemma~\ref{lem:default_form}, a gadget $G(i,j,u,v)$ is not suppressed if and only if $(i,j,u,v) \in (H_p, \mathcal{V}_p)$.
    Moreover, these gadgets must be colored using Selection, Helper, Symmetry, and Consistency colors, and each color $d$ of this family must still be used $h'(d) = h(d) - W - 1$ times.
    We break the proof in the following claims, where $\mathcal{G}_p(i,j)$ are all the gadgets where $(i,j,u,v) \in (H_p, \mathcal{V}_p)$.
    
    \begin{claim}
        \label{clm:parity}
        In every list coloring of $G$ satisfying $h$, exactly one gadget of each $\mathcal{G}_p(i,j)$ is selected, each passed $G(i,j,w,z)$ has all of its $kZ$ even vertices colored with $y(i,j)$, and the selected $G(i,j,u,v)$ has all of its $2 + kZ$ odd vertices colored with $y(i,j)$.
    \end{claim}
    
    \begin{cproof}
        For the first statement, if no gadget was selected, $\sigma'(i,j)$ would be used $M > M - 1$ times; if more than one is selected, $\sigma'(i,j)$ does not meet the target.
        If $G(i,j,u,v)$ is selected then $a(i, j,u,v)$ and $b(i,j,u,v)$ are colored with $y(i,j)$. After removing these vertices, we are left with two even paths of which at most half of its vertices are colored with the same color.
        As such there are at most $2 + f_\down + (k-1)u_\down + f_\up + (k-1)u_\up = 2 + kZ$ vertices colored with $y(i, j)$ in this gadget; moreover this bound is achieved if and only if the odd vertices are colored with $y(i,j)$.
        For each passed $G(i,j,w,z)$, $a(i, j, w,z)$ and $b(i, j, w,z)$ are colored with $\sigma(i, j)$, otherwise the numerical target of $\sigma(i, j)$ cannot be met.
        After the removal of $a,b,r(i, j, w,z)$ we are again left with two even paths.
        Thus, at most $e_\down + (k-1)u_\down + e_\up + (k-1)u_\up = kZ$ vertices have color $y(i, j)$ in this gadget.
        To see that $y(i,j)$ must be used only for even vertices of $G(i,j,w,z)$ to meet this bound, note that, if this is not done, then $x(i,j)$ will never meet its bound, since exactly $kMZ - kZ$ vertices remain (after the coloring of $G(i,j,u,v)$ and the instance selector gadget) that can be colored with $x(i,j)$, which is precisely the number of vertices that must still be colored with $x(i,j)$ to meet its target.
    \end{cproof}
    
    \begin{claim}
        In every list coloring $\varphi$ of $G$, if $G(i,j,u,v)$ is selected, so is $G(j,i,v,u)$.
    \end{claim}
        
    \begin{cproof}
        Suppose $i < j$, and let $e = uv$, and $\conj{e} = vu$
        By Claim~\ref{clm:parity} we know that for a selected gadget every odd vertex is colored with $y(i, j)$, so each even vertex is colored with a non-Helper color.
        Note that color $\varepsilon'(i,j)$ is used $e_\down$ times in gadget $\overrightarrow{G}(i,j,u,v)$.
        Now, we need to select one backward gadget of $\mathcal{G}_p(j,i)$; suppose we select gadget $G(j,i,w,z)$, $wz = f \neq \conj{e}$.
        Again by Claim~\ref{clm:parity}, the number of occurrences of $\varepsilon'(i,j)$ is $f_\up$ times in $G(j,i,w,z)$, but we have that $e_\down + f_\up \neq Z$, a contradiction that $\varphi$ satisfies $h$.
    \end{cproof}

    \begin{claim}
        In every list coloring $\varphi$ of $G$, if $G(i,j,u,v)$ is selected and $e = uv$, then, for every $s \neq i$, the edge $wz$ of $H_p$ corresponding to the selected gadget $G(i, s, w, z)$ have $w = u$.
    \end{claim}

    \begin{proof}
        We divide our proof in two cases.
        First, suppose $i < j$ and $s < j$.
        By Claim~\ref{clm:parity} $\tau_i(s,j)$ is used in $u_\up$ vertices of $G(i,j,u,v)$.
        Now, note that the only possible gadget $G(i,s,w,z)$ that can be chosen such that $w \in V_i$
        satisfies $u_\up + w_\down = Z$ must have $w = u$, since $\tau_i(s,j)$ is used $w_\down$ times in gadget $G(i,s,w,z)$.
        The case $j < s$ is similar, but we replace $\tau_i$ with $\tau'_i$ in the analysis.
    \end{proof}
    
    The previous claims guarantee that: two edges incident to vertices of $V_i^p$ have the same endpoint in $V_i^p$ for every $i \in k$ and, consequently, for every pair of edges $uv \in E(V_i^p, V_j^p)$, $uw \in E(V_i^p, V_\ell^p)$, if $G(i,j,u,v)$ and $G(i,j,u,w)$ are selected, then $(j,\ell, w, v) \in (H_p, \mathcal{V}_p)$ and $G(j,\ell,w,v)$ is also selected.
\end{proof}

\begin{lemma}
    \label{lem:vc_bound}
    Graph $G$ has a vertex cover of size $\poly(n + \log_2 |\mathcal{H}|)$.
\end{lemma}

\begin{proof}
    Note that there are at most $\binom{k}{2}\binom{n}{2}$ edge gadgets since we have at most this many tuples of the form $(i,j,u,v)$ where $i,j \in [k]$ and $u,v \in [n]$.
    Moreover, for each choice gadget $\mathcal{C}(d)$, we have $|\mathcal{C}(d) \setminus I_d| = \log_2 |\mathcal{H}|$ and, between the primary choice gadget $\mathcal{C}(c)$ and any other $\mathcal{C}(d)$ the only edges are between the bits, i.e. the union of the independent sets of the choice gadgets is also an independent set $\mathcal{I}$.
    Finally, since we have $\frac{q-1}{2} \in \bigO{k^2}$ choice gadgets and $k \leq n$, $G \setminus \mathcal{I}$ has $\poly(n + \log_2 |\mathcal{H}|)$ vertices and is a vertex cover of $G$.
\end{proof}

Lemmas~\ref{lem:forward_nlc_kernel},~\ref{lem:backward_nlc_kernel}, and~\ref{lem:vc_bound} directly imply Theorem~\ref{thm:no_kernel_nlc}: the first two lemmas show that the some $(H_p, \mathcal{V}_p) \in \mathcal{H}$ admits a solution if and only if $(G, L, h)$ admit a solution, while the third shows that the parameters in $(G, L, h)$ are bounded by a polynomial on $n + \log_2 |\mathcal{H}|$.

\begin{theorem}
    \label{thm:no_kernel_nlc}
    \pname{Number List Coloring} does not admit a polynomial kernel when parameterized by vertex cover and number of colors, unless $\NP \subseteq \coNP/\poly$.
\end{theorem}

To obtain the next two corollaries, note that: (i) adding isolated vertices with singleton lists to $G$ is a PPT reduction from \pname{Number List Coloring} parameterized by vertex cover and number of colors to \pname{Equitable List Coloring} under the same parameterization, and (ii) adding a clique $R$ of same size as the number of colors and edges between $R$ and vertices of $G$ to simulate the lists is also a PPT reduction from \pname{Equitable List Coloring} to \pname{Equitable Coloring}; note that the clique may be safely added to the vertex cover, since its size is linear on the number of colors, which is a parameter.

\begin{corollary}
    \pname{Equitable List Coloring} does not admit a polynomial kernel when parameterized by vertex cover and number of colors, unless $\NP \subseteq \coNP/\poly$.
\end{corollary}

\begin{corollary}
    \pname{Equitable Coloring} does not admit a polynomial kernel when parameterized by vertex cover and number of colors, unless $\NP \subseteq \coNP/\poly$.
\end{corollary}
    \section{Final Remarks}

In this work we presented an extensive study of multiple parameterizations for the \pname{Equitable Coloring} problem, obtaining both tractability, intractability, and kernelization results.
Specifically, we proved that it is fixed parameter tractable when parameterized by distance to cluster and by distance to co-cluster.
As corollaries, we have that there is an \FPT\ algorithm when parameterizing the problem by distance to unipolar and number of colors.
Meanwhile, \pname{Equitable Coloring} remains \WH\ when simultaneously parameterized by distance to disjoint paths and number of colors; if, on the other hand, we parameterize by distance to disjoint paths of length at most $\ell$, the problem is fixed-parameter tractable.
We then presented a linear kernel under distance to clique and a cubic kernel when parameterized by the maximum leaf number; along with Reddy's~\cite{equitable_threshold} polynomial kernel when parameterized by distance to connected threshold and number of colors, these are the only examples of positive kernelization results we know of.
Our final result showed that, under a standard complexity hypothesis,\pname{Equitable Coloring} has no polynomial kernel when jointly parameterized by vertex cover and number of colors.
We also revisited previous works in the literature and restated them in terms of parameterized complexity.
This review settled the complexity for some parameterizations weaker than distance to clique and show that some of our results are, in a sense, optimal: searching for parameters weaker than distance to (co-)cluster will most likely not yield \FPT\ algorithms.
\pname{Vertex Coloring} is already notoriously hard to find polynomial kernels for, as shown by Jansen and Kratsch~\cite{data_reduction}; in fact, most of the parameterizations under which classical coloring admits a polynomial kernel do not make \pname{Equitable Coloring} tractable, painting a quite bleak future for kernelization algorithms for \pname{Equitable Coloring}.
We leave three main questions about the tractability of \pname{Equitable Coloring}: distance to disconnected threshold, feedback edge set, and the stronger joint parameterization feedback vertex set and maximum degree, with the first two being the most interesting cases.
We believe it is also possible to improve the constants on our max leaf number kernel result, but do not yet know how to obtain a subcubic kernel under the this parameterization. 
Another interesting direction would be to improve the running times of known tractable cases and determining lower bounds for parameters such as vertex cover and distance to clique.
In particular, we would like to know if the algorithms presented in Sections~\ref{sec:dc},~\ref{sec:dcc}, and~\ref{sec:bpath} have optimal running times or, if not, how one can obtain faster algorithms.

    \bibliographystyle{abbrv}
    \bibliography{main}
\end{document}